\documentclass{article}

\usepackage{microtype}
\usepackage{graphicx}
\usepackage{booktabs} 
\usepackage{bm}
\usepackage{amsmath}
\usepackage{amsthm}
\usepackage{amsfonts}
\usepackage{caption}
\usepackage{subcaption}
\usepackage[textsize=scriptsize]{todonotes}

\usepackage{hyperref}

\usepackage[accepted]{icml2021}

\usepackage[normalem]{ulem}

\icmltitlerunning{Global Rhythm Style Transfer Without Text Transcriptions}

\def \algname {\textsc{AutoPST}}
\newcommand{\e}[1]{{\small $#1$}}

\newtheorem{theorem}{Theorem}

\begin{document}

\twocolumn[
\icmltitle{Global Rhythm Style Transfer Without Text Transcriptions}




\icmlsetsymbol{equal}{*}

\begin{icmlauthorlist}
\icmlauthor{Kaizhi Qian}{equal,ibm,ibm2}
\icmlauthor{Yang Zhang}{equal,ibm,ibm2}
\icmlauthor{Shiyu Chang}{ibm,ibm2}
\icmlauthor{Jinjun Xiong}{ibm2}
\icmlauthor{Chuang Gan}{ibm,ibm2}
\icmlauthor{David Cox}{ibm,ibm2}
\icmlauthor{Mark Hasegawa-Johnson}{uiuc}
\end{icmlauthorlist}

\icmlaffiliation{uiuc}{University of Illinois at Urbana-Champaign, USA}
\icmlaffiliation{ibm}{MIT-IBM Watson AI Lab, USA}
\icmlaffiliation{ibm2}{IBM Thomas J. Watson Research Center, USA}

\icmlcorrespondingauthor{Yang Zhang}{yang.zhang2@ibm.com}
\icmlcorrespondingauthor{Kaizhi Qian}{kaizhiqian@gmail.com }

\icmlkeywords{Machine Learning, ICML}

\vskip 0.3in
]



\printAffiliationsAndNotice{\icmlEqualContribution} 

\begin{abstract}

Prosody plays an important role in characterizing the style of a speaker or an emotion, but most non-parallel voice or emotion style transfer algorithms do not convert any prosody information. Two major components of prosody are pitch and rhythm.  Disentangling the prosody information, particularly the rhythm component, from the speech is challenging because it involves breaking the synchrony between the input speech and the disentangled speech representation.  As a result, most existing prosody style transfer algorithms would need to rely on some form of text transcriptions to identify the content information, which confines their application to high-resource languages only. Recently, \textsc{SpeechSplit} \cite{qian2020unsupervised} has made sizeable progress towards unsupervised prosody style transfer, but it is unable to extract high-level global prosody style in an unsupervised manner.  In this paper, we propose \algname, which can disentangle global prosody style from speech without relying on any text transcriptions.  \algname\ is an Autoencoder-based Prosody Style Transfer framework with a thorough rhythm removal module guided by self-expressive representation learning. Experiments on different style transfer tasks show that \algname\ can effectively convert prosody that correctly reflects the styles of the target domains.

\end{abstract}



\section{Introduction}
\label{sec:intro}

Speech contains many layers of information. Besides the speech content, which can roughly be transcribed to text for many languages, prosody also conveys rich information about the personal, conversational, and world context within which the speaker expresses the content. There are two major constituents of prosody. The first is \emph{rhythm}, which summarizes the sequence of phone durations, and expresses phrasing, speech rate, pausing, and some aspects of prominence. The second is \emph{pitch}, which reflects 
intonation.

Recently, non-parallel speech style transfer tasks have achieved rapid progress thanks to the advancement of non-parallel deep style transfer algorithms. Speech style transfer refers to the tasks of transferring the source speech into the style of the target domain, while keeping the content unchanged. For example, in voice style transfer, the domains correspond to the speaker identities. In emotion style transfer, the domains correspond to the emotion categories. In both of these tasks, prosody is supposed to be an important part of the domain style --- different speakers or emotions have distinctive prosody patterns. However, few of the state-of-the-art algorithms in these two applications can convert the prosody aspect at all. Typically, the converted speech would almost always maintain the same pace and pitch contour shape as the source speech, even if the target speaker or emotion has a completely different prosody style.

The fundamental cause of not converting prosody is that disentangling the prosody information, particularly the rhythm aspect, is very challenging. Since the rhythm information corresponds to how long the speaker utters each phoneme, deriving a speech representation with the rhythm information removed implies breaking the temporal synchrony between the speech utterance and the representation, which has been shown difficult \citep{watanabe2017hybrid, kim2017joint} even for supervised tasks (\emph{e.g.} automatic speech recognition (ASR)), using state-of-the-art asynchronous sequence-to-sequence architectures such as Transformers \citep{vaswani2017attention}.

Due to this challenge, most existing prosody style transfer algorithms are forced to use text transcriptions to identify the content information in speech, and thereby separate out the remaining information as style \cite{biadsy2019parrotron}. However, such methods are language-dependent and cannot be applied to low-resource languages with few text transcriptions. Although there are some sporadic attempts to disentangle prosody in an unsupervised manner \cite{polyak2019attention}, their performance is limited. These algorithms typically consist of an auto-encoder pipeline with a resampling module at the input to corrupt the rhythm information. However, the corruption is often so mild that most rhythm information can still get through. \textsc{SpeechSplit} \cite{qian2020unsupervised} can achieve a more thorough prosody disentanglement, but it needs access to a set of ground-truth fine-grained local prosody information in the target domain, such as the exact timing of each phone and exact pitch contour, which is unavailable in the aforementioned style transfer tasks that can only provide high-level global prosody information. In short, global prosody style transfer without relying on text transcriptions or local prosody ground truth largely remains unresolved in the research community.

Motivated by this, we propose \algname, an unsupervised speech decomposition algorithm that 1) does not require text annotations, and 2) can effectively convert prosody style given domain summaries (\emph{e.g.} speaker identities and emotion categories) that only provide high-level global information. As its name indicates, \algname\ is an Autoencoder-based Prosody Style Transfer framework. \algname\ introduces a much more thorough rhythm removal module guided by the self-expressive representation learning proposed by \citet{Bhati2020}, and adopts a two-stage training strategy to guarantee passing full content information without leaking rhythm.  Experiments on different style transfer tasks show that \algname\ can effectively convert prosody that correctly reflects the styles of the target domains.

\section{Related Work}


\textbf{Prosody Disentanglement}\quad Several prosody disentanglement techniques are found in expressive text-to-speech (TTS) systems.
~\citet{skerry2018towards} introduced a Tacotron based speech synthesizer that can disentangle prosody from speech content by an auto-encoder based representation. \citet{wang2018style} further extracts global styles by quantization. Mellotron \citep{valle2020mellotron} is a speech synthesizer that captures and disentangles different aspects of the prosody information. CHiVE \citep{kenter2019chive} explicitly extracts and utilize prosodic features and linguistic features for expressive TTS.
However, these TTS systems all require text transcriptions, which, as discussed, makes the task easier but limits their applications to high-resource language.
Besides TTS systems, Parrotron \cite{biadsy2019parrotron} disentangles prosody by encouraging the latent codes to be the same as the corresponding phone representation of the input speech. \citet{liu2020towards} proposed to disentangle phoneme repetitions by vector quantization. However, these systems still require text transcriptions. \citet{polyak2019attention} proposed a prosody disentanglement algorithm that does not rely on text transcriptions, which attempts to remove the rhythm information by randomly resampling the input speech. However, the effect of their prosody conversion is not very pronounced. \textsc{SpeechSplit} \citep{qian2020unsupervised} can disentangle prosody with better performance, but it relies on fine-grained prosody ground-truth in the target domain.

\textbf{Voice Style Transfer}\quad Many style transfer approaches have been proposed for voice conversion. VAE-VC \cite{hsu2016voice} and VAE-GAN \cite{hsu2017voice} directly learns speaker-independent content representations using a VAE. ACVAE-VC \cite{kameoka2018acvae} encourages the converted speech to be correctly classified as the target speaker by classifying the output. In contrast, \citet{chou2018multi} discourages the latent code from being correctly classified as the source speaker by classifying the latent code. Inspired by image style transfer, \citet{gao2018voice} and \citet{kameoka2018stargan} adapted CycleGan \citep{kaneko2017parallel} and StarGan \citep{choi2018stargan} respectively for voice conversion. Later, CDVAE-VC was extended by directly applying GAN \cite{huang2020unsupervised} to improve the degree of disentanglement. \citet{chou2019one} uses instance normalization to further disentangle timbre from content. StarGan-VC2 \cite{kaneko2019stargan} refines the adversarial network by conditioning the generator and discriminator on both the source and the target speaker. \textsc{AutoVC} \cite{qian2019autovc} disentangles the timbre and content by tuning the information-constraining bottleneck of a simple autoencoder. Later, \citet{qian2020f0} fixed the pitch jump problem of \textsc{AutoVC} by F0 conditioning. Besides, the time-domain deep generative model is also gaining popularity \cite{niwa2018statistical,nachmani2019unsupervised,serra2019blow}. However, these methods only focus on converting timbre, which is only one of the speech components.

\textbf{Emotion Style Transfer}\quad Most existing emotion style transfer algorithms disentangle prosody information using parallel data. Early methods \cite{tao2006prosody, wu2009hierarchical} use classification and regression tree to disentangle prosody. Later, statistical methods, such as GMM \cite{aihara2012gmm} and HMM \cite{inanoglu2009data}, and deep learning methods \cite{luo2016emotional, ming2016deep} are applied. However, these approaches use parallel data. Recently, non-parallel style transfer methods \cite{zhou2020transforming, zhou2020vaw} are applied to emotion style transfer to disentangle prosody. However, these methods are unable to explicitly disentangle rhythm information.

\section{The Challenges of Disentangling Rhythm}
In this section, we will provide some background information on why disentangling rhythm has been difficult.

\begin{figure}[t]
\centering
\includegraphics[width=0.9\linewidth]{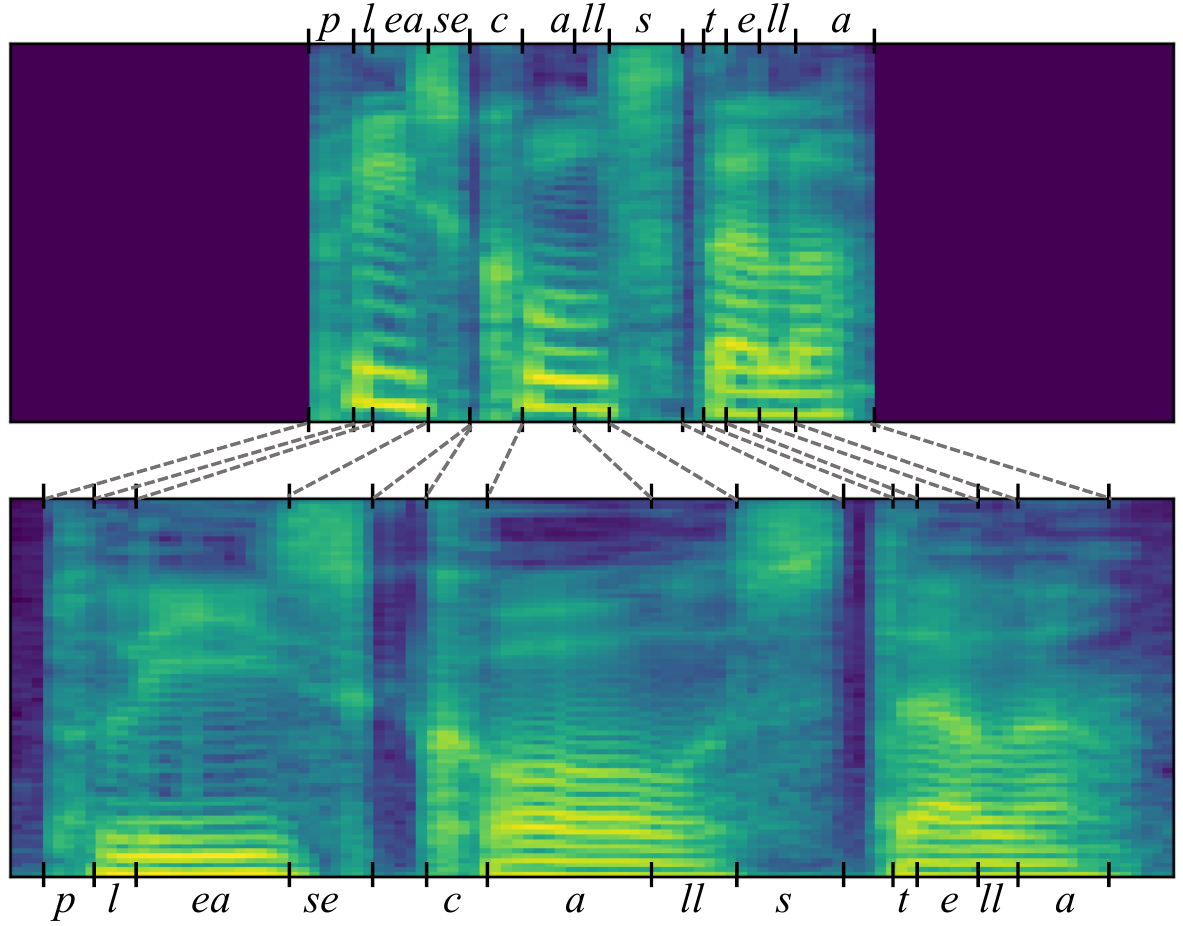}
\caption{Mel-spectrograms of two utterances of \emph{``Please call Stella''} with different speech rates. The phone segments are marked on the $x$-axis. The grey dashed lines between the two spectrograms mark the phone alignment, which shows that the changes in duration are disproportional across different phones.}
\label{fig:sepct}
\end{figure}

\subsection{Rhythm Information}

Figure~\ref{fig:sepct} shows two utterances of the same word \emph{``Please call Stella''}, with the phone segments marked on the $x$-axis. As shown, although the content is almost the same across the two utterances, each phone in the second utterance is repeated for more frames than in the first utterance. In this paper, we will measure rhythm as the number of frame-aligned repetitions of each phone symbol. 

Formally, denote \e{\bm X(t)} as the spectrogram of a speech utterance, where \e{t} represents the frame index. Denote \e{\bm S} as the vector of phone symbols contained in an utterance, which we also refer to as the content information. Denote \e{\bm R} as a vector of the number of repetitions of each phone symbol in \e{\bm S}, which we also refer to as the rhythm information. Then \e{\bm S} and \e{\bm R} represent two information components in \e{\bm X(t)}. Therefore, the task of disentangling rhythm involves removing phone repetitions while retaining phone identities. Formally, we would like to derive a hidden representation \e{\tilde{\bm Z}} from the speech \e{\bm X}, \emph{i.e.} \e{\tilde{\bm Z} = f(\bm X)}, according to the following objective
\begin{equation}
\small
    \min_{\tilde{\bm Z} = f(\bm X)} I(\bm R; \tilde{\bm Z}),
\end{equation}
\vspace{-0.5em}
\begin{equation}
    \small
    \mbox{s.t. } I(\bm S; \tilde{\bm Z}) = I(\bm S; \bm X),
    \label{eq:constraint}
\end{equation}
where \e{I(\cdot; \cdot)} denotes mutual information. If the text transcriptions were available, this objective could be achieved by training an asynchronous sequence-to-sequence model, such as the Transformer \cite{vaswani2017attention}, to predict the phone sequence from input speech. However, without the phonetic labels, it remains an unresolved challenge to uncover the phonetic units in an unsupervised manner.

\subsection{Disentangling Rhythm by Resampling}
\label{subsec:segment}

Existing unsupervised rhythm disentanglement techniques seek to obscure the rhythm information by temporally resampling the speech sequence or its hidden representation sequence. One common resampling strategy, as proposed by \citet{polyak2019attention}, involves three steps. First, the input sequence is separated into segments of random lengths. Second, for each segment, a sampling rate is randomly drawn from some distribution. Finally, the corresponding segment is resampled at this sampling rate.

To gauge how much rhythm information the random resampling can remove, consider two hypothetical speech sequences, \e{\bm X(t)} and \e{\bm X'(t)}, with the same content information \e{\bm S}, but with different rhythm information \e{\bm R} and \e{\bm R'} respectively. If random resampling is to reduce the 
rhythm information in each utterance,
there must be a non-zero probability that the random resampling temporally aligns them (hence making them identical because their only difference is rhythm). Thus, the original rhythm information distinguishing the two utterances is removed. Formally, denote \e{\tilde{\bm Z}(t)} and \e{\tilde{\bm Z}'(t)} as the resampled outputs of the two unaligned but otherwise identical sequences. Then a necessary condition for reduction in rhythm information is that
\begin{equation}
\small
    Pr(\tilde{\bm Z}(t) = \tilde{\bm Z}'(t)) > 0.
    \label{eq:nes_cond}
\end{equation}
Higher probability removes more rhythm information. If the probability is zero, then \e{I(\bm R; \tilde{\bm Z})} will reach its upper bound, which is \e{I(\bm R; \bm X)}. If the probability is one, then \e{I(\bm R; \tilde{\bm Z})} would reach its lower bound, which is \e{I(\bm R; \bm S)}. Please refer to Appendix~\ref{append:theory} for a more formal analysis.

Therefore, we argue that the above random resampling algorithm does not remove much rhythm information, because it has a low probability of aligning any utterance pairs with the same content but different rhythm patterns. Figure~\ref{fig:sepct} is a counterexample, where the duration of each phone changes disproportionally. The vowel ``ea'' gets more stretched than the consonants. Consider the case where the entire sequence falls into one random segment in the first step of the resampling algorithm. Then, due to the disproportional stretch among the phones, it is impossible to simultaneously align all the phones by uniformly stretching or squeezing the two utterances. If the utterances are broken into multiple random segments, it is possible to achieve the alignment, but this requires the number of random segments is greater than or equal to the number of phones, with at least one segment boundary between each pair of phone boundaries, whose probability of occurring is extremely low.
In short, 
to overcome the limitation of the existing resampling scheme, 
we need better ways to account for the disproportional variations in duration across different phones.

\section{The \algname\ Algorithm}
\begin{figure}[t]
\centering
\includegraphics[width=0.9\linewidth]{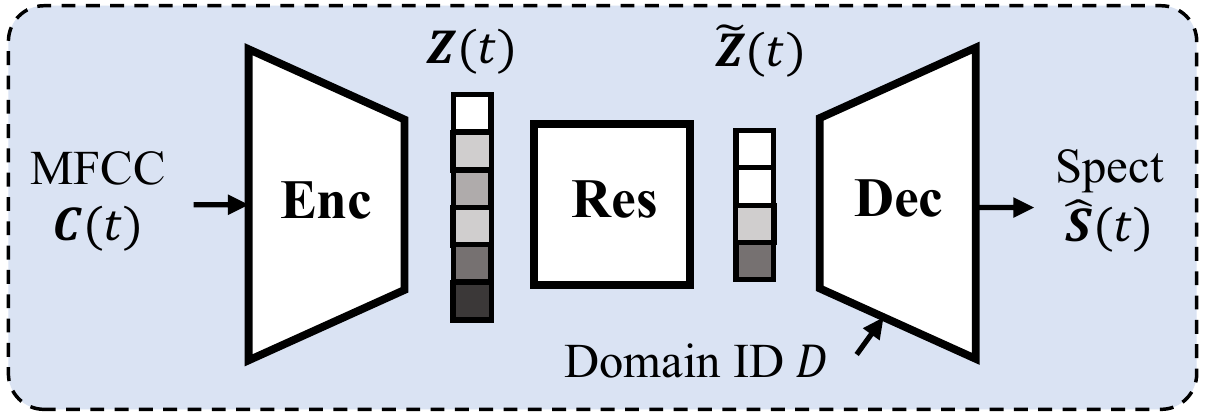}
\caption{The framework of \algname. ``Enc'', ``Res'' and ``Dec'' denote the encoder, the resampling module, and the decoder respectively; ``Spect'' represents the spectrogram. Each block in \e{\bm Z(t)} and \e{\tilde{\bm Z}(t)} represents a frame. The blocks with similar color shades denote that the corresponding frames have a high similarity.}
\label{fig:model}
\end{figure}

In this section, we will introduce \algname, which can effectively overcome the aforementioned limitations. 

\subsection{The Framework Overview}
As shown in Figure~\ref{fig:model}, \algname\ adopts a similar autoencoder structure as \textsc{AutoVC} \cite{qian2019autovc}.
The decoder aims to reconstruct speech based on the output of the random resampling module and the domain identifiers, \emph{e.g.}, the speaker identity. \textsc{AutoVC} has been shown effective in disentangling the speaker's voice via its information bottleneck, but it does not disentangle pitch and rhythm.

Therefore, \algname\ introduces three changes. First, instead of spectrogram, \algname\ takes the 13-dimensional MFCC having little pitch information. Second, \algname\ introduces a novel resampling module 
guided by self-expressive representation learning \cite{Bhati2020}, which can overcome the challenges mentioned in the previous section. Finally, \algname\ adopts a two-stage training scheme to prevent leaking rhythm information.

Formally, denote \e{\bm X(t)} as the speech spectrogram, \e{\bm C(t)} as the input MFCC feature, and \e{D} as the domain identifier. The reconstruction process is described as follows
\begin{equation}
\small
\begin{aligned}
    &\bm Z(t) = \mathrm{Enc}(\bm C(t)), ~ \tilde{\bm Z}(t) = \mathrm{Res}(\bm Z(t)), \\
    & \hat{\bm X}(t) = \mathrm{Dec}(\tilde{\bm Z}(t), D) \longleftrightarrow \bm X(t)
\end{aligned}
\label{eq:framework}
\end{equation}
where \e{\mathrm{Enc}}, \e{\mathrm{Res}}, \e{\mathrm{Dec}} stand for the encoder, the resampling module and the decoder, respectively. Sections~\ref{subsec:downsample} to \ref{subsec:upsample} will introduce the random resampling module, and Section~\ref{subsec:training} will discuss the training scheme.

\begin{figure}[t]
	\centering
	\begin{subfigure}{\columnwidth}
		\centering
		\includegraphics[width=0.9\linewidth]{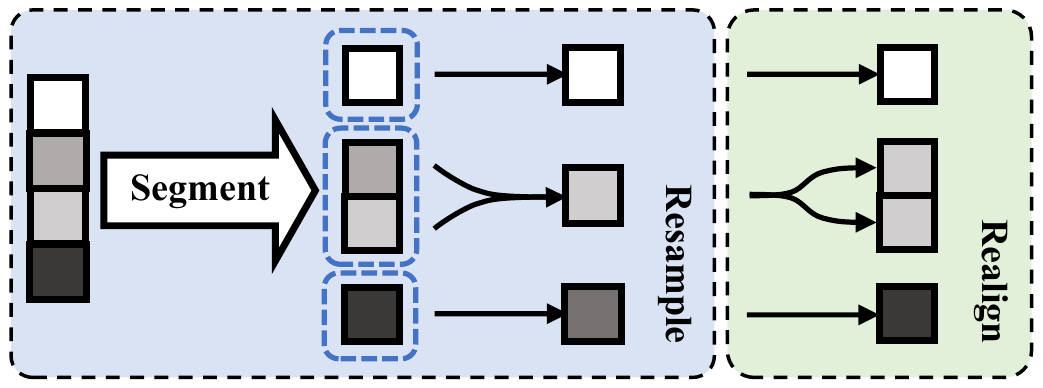}
		\caption{The downsampling case ($\tau \leq 1$)}\label{subfig:downsample}
		\vspace{0.05in}
	\end{subfigure}
	\begin{subfigure}{\columnwidth}
		\centering
		\includegraphics[width=0.9\columnwidth]{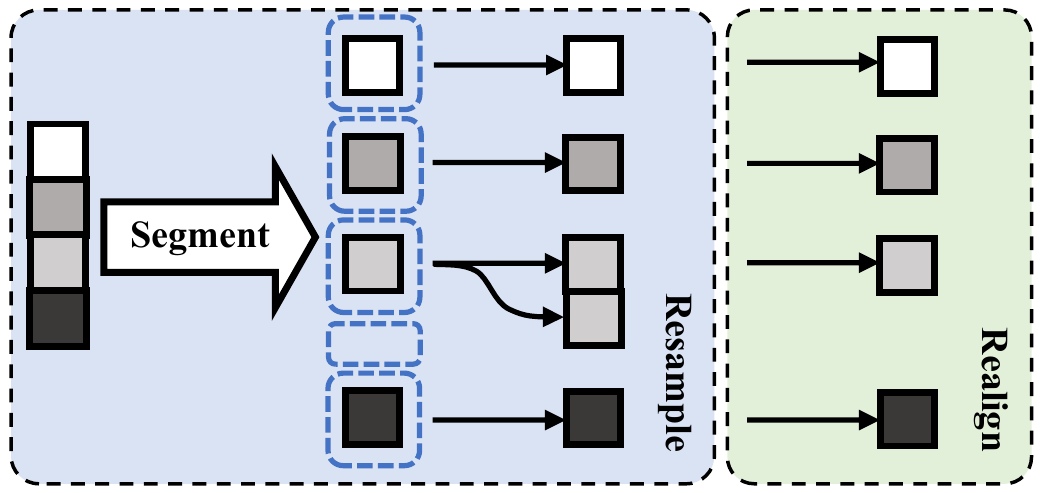}
		\caption{The upsampling case ($\tau > 1$)}\label{subfig:upsample}
	\end{subfigure}
	\caption{The resampling module (left) and realignment module (right) of \algname. Merging arrows denote mean-pooling. Splitting arrows denote copying the input into multiple replicates. The blocks with similar color shades denote that the corresponding frames have a high similarity. According to the shade, frames 2 and 3 are similar, while the others frames are very dissimilar. (a) When $\tau \leq 1$, the sequence is segmented based on similarity, and each segment is merged to one code by mean-pooling. (b) When $\tau > 1$, each segment contains only one code. In addition, empty segments are inserted where the inter-temporal similarity
	is high, whose corresponding output positions replicate the previous codes.}
	\label{fig:resample}
\end{figure}

\begin{figure}[t]
\centering
\includegraphics[width=0.9\linewidth]{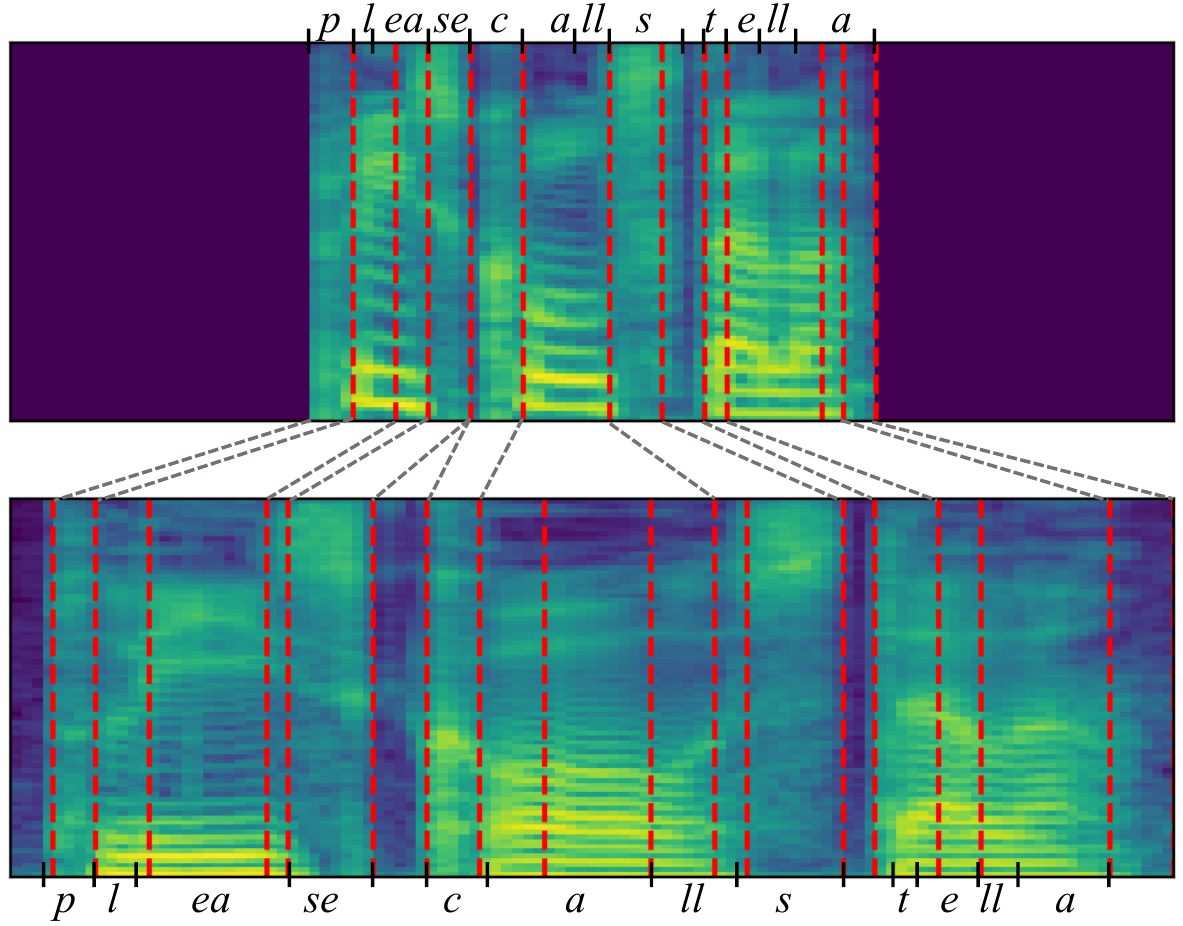}
\caption{Segmentation results of the downsampling algorithm of the two utterances in Figure~\ref{fig:sepct}. The vertical dashed lines represent the segment boundaries. The dashed lines between the two spectrograms show good content alignment among the segments.}
\label{fig:spect_seg}
\end{figure}

\begin{figure*}[t]
\raggedright
\includegraphics[width=1\linewidth]{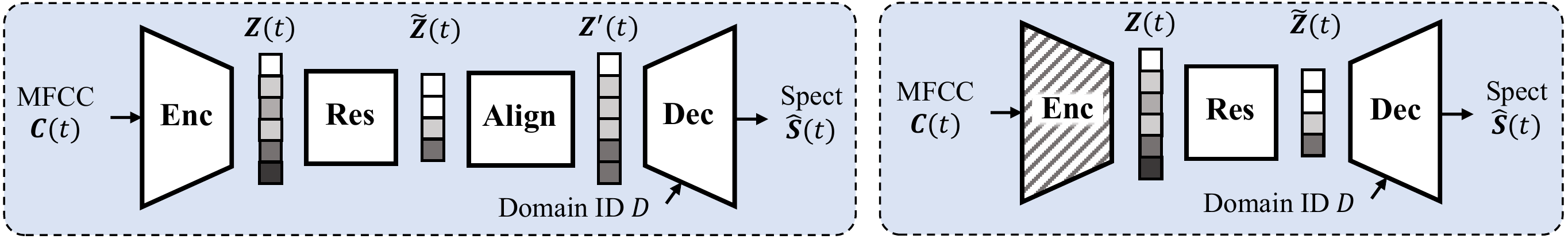}
\\
\hspace{1.2in}\small{(a) Synchronous training.\label{subfig:sync}} \hspace{2in} \small{(b) Asynchronous training.\label{subfig:async}}
\caption{The two-stage training scheme of \algname. ``Align'' denotes the re-alignment module. The striped encoder in (b) means that the parameters in it are frozen. The blocks with similar color shades denote that the corresponding frames have a high similarity.}
\label{fig:training}
\end{figure*}

\subsection{Similarity-Based Downsampling}
\label{subsec:downsample}

Our resampling scheme capitalizes on the observation
that the relatively steady segments in speech tend to have more flexible durations, such as the ``ea'' segment in Figure~\ref{fig:sepct}. We thus modify the self-expressive autoencoder (SEA) algorithm proposed by \citet{Bhati2020} into a similarity-based downsampling scheme. SEA derives a frame-level speech representation, which we denote as \e{\bm A(t)}, that contrastively promotes a high cosine similarity between frames that are similar, and a low cosine similarity between dissimilar frames. We then create a Gram matrix, \e{\mathsf{G}} to record the cosine similarities between any frame pairs:
\begin{equation}
    \small
    \mathsf{G}(t, t') = \frac{\bm A^T(t) \bm A(t')}{\Vert \bm A(t)\Vert_2 \Vert \bm A(t')\Vert_2}.
\end{equation}
More details of the SEA algorithm can be found in Appendix~\ref{subsec:sea} and the original paper \cite{Bhati2020}.

As shown in the left panel of Figure~\ref{fig:resample}(\subref{subfig:downsample}), our downsampling scheme for \e{\bm Z(t)} involves two steps. First, we break \e{\bm Z(t)} into consecutive segments, such that the cosine similarity of \e{\bm A(t)} are high within each segment, and that the cosine similarity drop across the segment boundaries. Second, each segment is merged into one code by mean-pooling.

Formally, denote the \e{t_m} as the left boundary for the $m$-th segment. Boundaries are sequentially determined. When all the boundaries up to \e{t_m} are determined, the next boundary \e{t_{m+1}} is set to \e{t} if \e{t} is the smallest time in \e{(t_m, \infty)} where the cosine similarity between \e{t} and \e{t_m} drops below a threshold:
\begin{equation}
\small
    \forall t' \in [t : t + 1], \quad \mathsf{G}(t_m, t')) \leq \tau(t).
\label{eq:segment}
\end{equation}
\e{\tau(t)} is a pre-defined threshold that can vary across \e{t}. Section~\ref{subsec:threshold} will discuss how to set the threshold. After all the segments are determined, each segment is reduced to one code by mean pooling, \emph{i.e.}
\begin{equation}
\small
    \tilde{\bm Z}(m) = \mathrm{meanpool}\big(\bm Z(t_m : t_{m+1}-1)\big).
    \label{eq:meanpool}
\end{equation}

Figure~\ref{fig:resample}(\subref{subfig:downsample}) shows a toy example, where the input sequence of length four. The second and the third codes are very similar. Then with a proper choice of \e{\tau(t)}, the downsampling would divide the input sequence into three segments, and collapse each segment into one code by mean-pooling. Note that the threshold \e{\tau(t)} governs how tolerant the algorithm is to dissimilarities. If \e{\tau(t)=1}, each code will be assigned to an individual segment, leading to no length reduction.

Figure~\ref{fig:spect_seg} shows the segmentation result of the two utterances shown in Figure~\ref{fig:sepct}, where the vertical dashed lines denote the segment boundaries. We can see that despite their significant difference in length, the two utterances are broken into approximately equal number of segments and the segments have a high correspondence in terms of content. Since the downsampled output is obtained by mean-pooling each segment, we can expect that their downsampled output would be very similar and temporally-aligned, which implies that the necessary condition for rhythm information loss (Equation~\eqref{eq:nes_cond}) is approximately satisfied.

\begin{figure*}[h]
\centering
\includegraphics[width=0.95\linewidth]{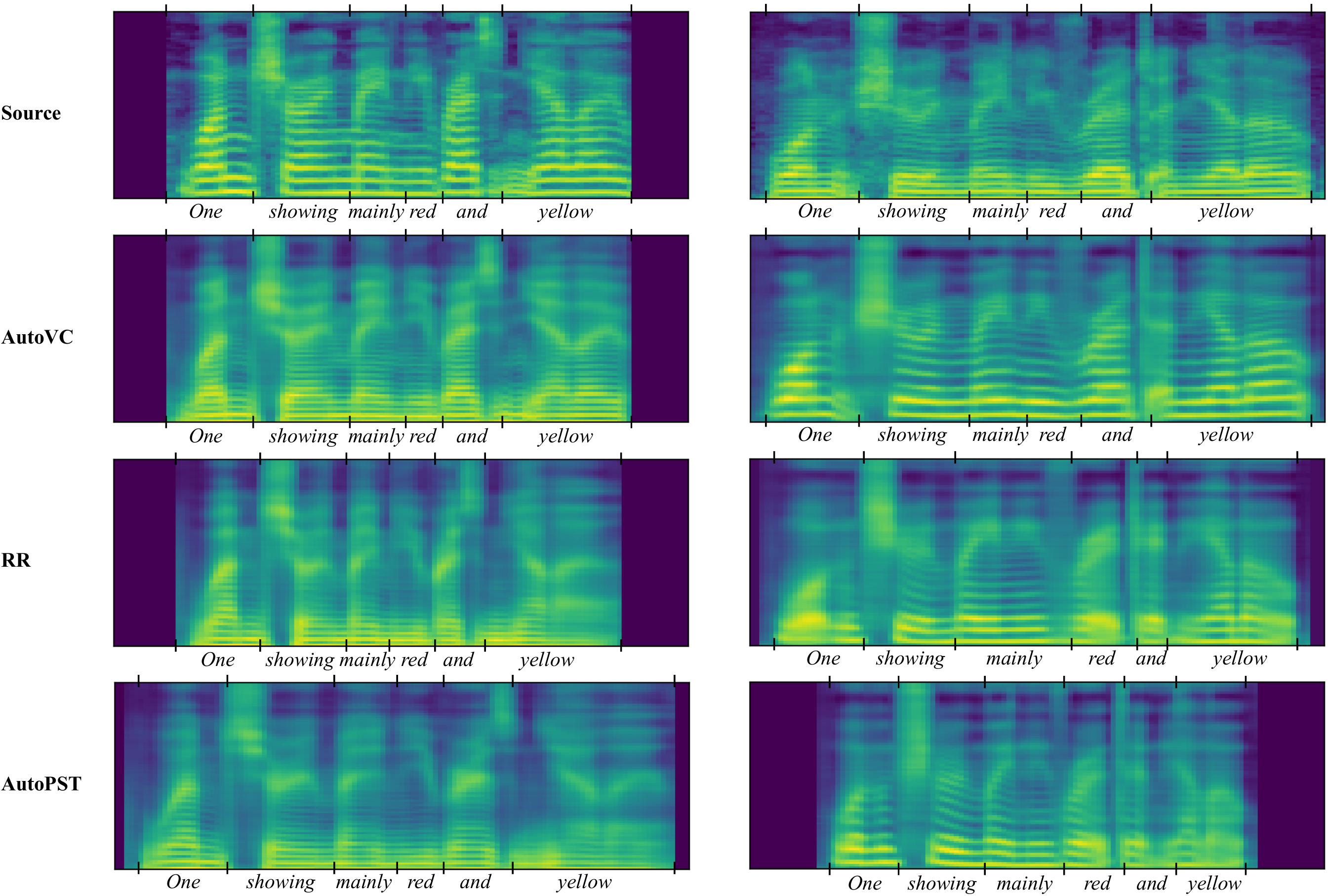}
\\
\raggedright
\hspace{1.3in}\small{(a) Fast-to-slow conversion.} \hspace{1.7in} \small{(b) Slow-to-fast conversion.}
\caption{Mel-spectrograms of the converted utterances of \textit{``One showing mainly red and yellow''} between a fast speaker and a slow speaker. \algname\ is the only algorithm that can significantly change the rhythm to match the target speakers' styles.
}
\label{fig:spect_convert}
\end{figure*}

\subsection{Randomized Thresholding}
\label{subsec:threshold}

For any fixed threshold \e{\tau} in Equation~\eqref{eq:segment}, there is a trade-off between rhythm disentanglement and content loss. The lower the \e{\tau}, the more rhythm information is removed, but the more content is lost as well. Ideally, during testing, we would like to set the threshold to $1$ 
to pass full content information to \e{\tilde{\bm Z}(t)}, and make the decoder ignore all the rhythm information in \e{\tilde{\bm Z}(t)}. This can be achieved with a randomized thresholding rule.

To see why, notice that if the decoder were to use the rhythm information in \e{\tilde{\bm Z}(t)}, it must know the value of \e{\tau}, because how the decoder (partially) recovers the rhythm information depends on how the rhythm information is collapsed, which is governed by \e{\tau}. However, the large variations in speech rate, utterance length and rhythm patterns in the training speech would overshadow the variations in \e{\tau}, making it extremely hard to estimate the value of \e{\tau}. Thus, the decoder will ignore whatever rhythm information remains in \e{\tilde{\bm Z}(t)}.

We adopt a double-randomized thresholding scheme. We first randomly draw a global variable \e{G \sim \mathcal{U}[u_l, u_r]} that is shared across the entire utterance, where \e{\mathcal{U}[u_l, u_r]} denotes the uniform distribution within the interval \e{[u_l, u_r]}. 
Then to determine if time \e{t} should be the next segment boundary (\emph{i.e.}, \e{t_{m+1}} in Equation~\eqref{eq:segment}), we draw a local variable \e{L(t) \sim \mathcal{U}[G-0.05, G+0.05]}. Then
\begin{equation}
\small
    \tau(t) = L(t)\mathrm{-quantile} \big[\mathsf{G}(t_m, t_m - b : t_m + b)\big].
    \label{eq:threshold_dist}
\end{equation}
\e{q\mathrm{-quantile}[\cdot]} denotes taking the \e{q}-quantile, and b denotes the length of the sliding window within which the threshold is computed, which is set to 20 in our implementation.

The motivation for setting the two levels of randomization is that \e{G} can obscure the global speech rate information and \e{L(t)} can obscure the local fine-grained rhythm patterns.

\subsection{Similarity-Based Upsampling}
\label{subsec:upsample}

To further obscure the rhythm information, we generalize our resampling module to accommodate upsampling. Just as downsampling aims to mostly shorten segments with higher similarity (hence decreasing the disproportionality), upsampling aims to mostly lengthen segments with higher similarity (hence increasing the disproportionality).

In the downsampling case, \e{\tau=1} implies no length reduction at all. We thus seek to extrapolate the case to \e{\tau > 1}, where the higher the \e{\tau}, the more the sequence gets lengthened. Our upsampling algorithm achieves this by inserting new codes in between adjacent codes. Specifically, suppose all the boundaries up to \e{t_m} are determined. When \e{\tau(t) > 1}, according to Equation~\eqref{eq:segment}, \e{t_{m+1}} will definitely be set to \e{t}. In addition to this, we will add yet another sentence boundary to \e{t}, \emph{i.e.} \e{t_{m+2} = t}, if
\begin{equation}
\small
    \forall t' \in [t : t + 1], \quad \mathsf{G}(t_m, t') \geq 1-\tau(t).
\end{equation}
In other words, we are inserting an empty segment for the \e{(m+1)}-th segment (because \e{t_{m+1} = t_{m+2}}). During the mean-pooling stage, this empty segment will be mapped to the code at its left boundary, \emph{i.e.},
\begin{equation}
\small
    \tilde{\bm Z}(m) = \bm Z(t_m), \quad \mbox{if } t_m = t_{m+1}.
    \label{eq:meanpool2}
\end{equation}
The non-empty segments will still be mean-pooled the same way as in Equation~\eqref{eq:meanpool}. 

The left panel of Figure~\ref{fig:resample}(\subref{subfig:upsample}) illustrates the upsampling process with the length-four toy example. Similar to the case of \e{\tau=1}, all the codes are individually segmented. The difference is that a new empty segment is inserted after the third code, which is where the cosine similarity is very high. At the mean-pooling stage, this empty segment turns into an additional code that copies the previous code. 

\subsection{Summary of the Resampling Algorithm}

To sum up, our resampling algorithm goes as follows. 

$\bullet$\quad For each frame, a random threshold \e{\tau(t)} is drawn from the distribution specified in Equation~\eqref{eq:threshold_dist}. 

$\bullet$\quad If \e{\tau(t) < 1}, the current frame would be either merged into the previous segment, or start a new segment, depending on whether its similarity with the previous segment exceeds the \e{\tau}, hence achieving downsampling (as elaborated in Section~\ref{subsec:downsample}).

$\bullet$\quad If \e{\tau(t) \geq 1}, the current frame would form either one new segment or two new segments (by duplicating the current frame), depending on whether its similarity with the previous segment exceeds \e{1-\tau(t)}, hence achieving upsampling (as elaborated in Section~\ref{subsec:upsample}).

$\bullet$\quad Move onto the next frame and repeat the previous steps.

Because the threshold for each frame is random, an utterance could be downsampled at some parts, while upsampled at others. This would ensure the rhythm information is sufficiently scrambled. As a final remark, the random threshold distribution (Equation~\eqref{eq:threshold_dist}) is governed by the percentile of the similarity, because the percentile has a direct correspondence to the length of the utterance after resampling. Appendix~\ref{subsec:vis} provides a visualization of how different thresholds affect the length of the utterance after resampling.

\subsection{Two-Stage Training}
\label{subsec:training}

Despite the resampling module, it is still possible for the encoder and decoder to find alternative ways to communicate the rhythm information that is robust against temporal resampling. Thus we introduce a two-stage training scheme to prevent any possible collusion.

The first stage of training, called the \emph{synchronous training}, realigns \e{\tilde{\bm Z}(m)} with \e{\bm Z(m)}, 
as shown in the right panels of Figure~\ref{fig:resample}. Specifically, 
for the downsampling case, we copy each \e{\tilde{\bm Z}(m)} to match the length of the original segment from which the code is mean-pooled; for the upsampling case, we delete the newly inserted \e{\tilde{\bm Z}(m)}. 
The network is then trained end-to-end to reconstruct the input with the realignment module, as shown in Figure~\ref{fig:training}(a). 
Since the decoder has full access to the rhythm information, the encoder will be trained to pass the content information and not the rhythm information.
The second stage, called \emph{asynchronous training}, removes the realignment module, freezes the encoder, and only updates the decoder, as shown in Figure~\ref{fig:training}(b). 

\section{Experiments}
We evaluate \algname\ on speech style transfer tasks. 
Additional experiment results can be found in Appendix~\ref{append:exper_results}. We encourage readers to listen to our online demo audios\footnote{\url{https://auspicious3000.github.io/AutoPST-Demo}}.

\subsection{Configurations}

\textbf{Architecture}\quad The encoder consists of eight $1\times 5$ of convolution layers with group normalization \cite{wu2018group}. The encoder output dimension is set to four. The decoder is a Transformer with four encoder layers and four decoder layers. The spectrogram is converted back to waveform using a WaveNet vocoder \cite{oord2016wavenet}.  More hyper-parameters setting details can be found in Appendix~\ref{append:exper_detail}.

\textbf{Dataset}\quad Our dataset is VCTK \cite{veaux2016superseded}, which consists of 44 hours of speech from 109 speakers. We use this dataset to perform the voice style transfer task, so the domain ID is the speaker ID. We use 24 speakers for training and follow the same train/test partition as in \cite{qian2020unsupervised}. We select the two fastest speakers and two slowest speakers from the seen speakers for evaluating rhythm style transfer. For further evaluation, we select two other speaker pairs with smaller rhythm differences. The first pair consists of speakers whose speech rates are at 25\% and 75\% percentiles among all the test speakers; the second pair at 40\% and 60\%, respectively.
More details can be found in Appendix~\ref{append:exper_detail}.


\textbf{Baselines}\quad We introduce two baselines. The first is the F0-assisted \textsc{AutoVC} \cite{qian2020f0}, an autoencoder-based voice style transfer algorithm. For the second baseline, we replace the \algname's random resampling module with that of the \textsc{SpeechSplit} (as introduced in Section~\ref{subsec:segment}). We refer to this baseline as RR (random resample).

\subsection{Spectrogram Visualization}

For a fast-slow speaker pair and one of their parallel sentences, \textit{``One showing mainly red and yellow''}, Figure~\ref{fig:spect_convert} shows the spectrograms of conversions from the slow speaker to the fast (left panel), and from the fast speaker to the slow (right panel). The word alignment is marked on the $x$-axis. As shown, all the algorithms can change the voice to the target speaker, as indicated by the average F0 and formant frequencies. However, only \algname\  significantly changes the rhythm towards the desired direction. \textsc{AutoVC} and RR barely change the rhythm.
It is also worth noting that \algname\ indeed learns the disproportionality in duration changes across phones, most duration changes occur in the steady vowel segments, \emph{e.g.}, \textit{``ow''} in \textit{``yellow''}. This verifies that our similarity-based resampling scheme can effectively obscure the relative length of each phone.

\begin{figure}[!ht]
\centering
    \begin{subfigure}{\columnwidth}
        \centering
        \includegraphics[width=0.84\linewidth]{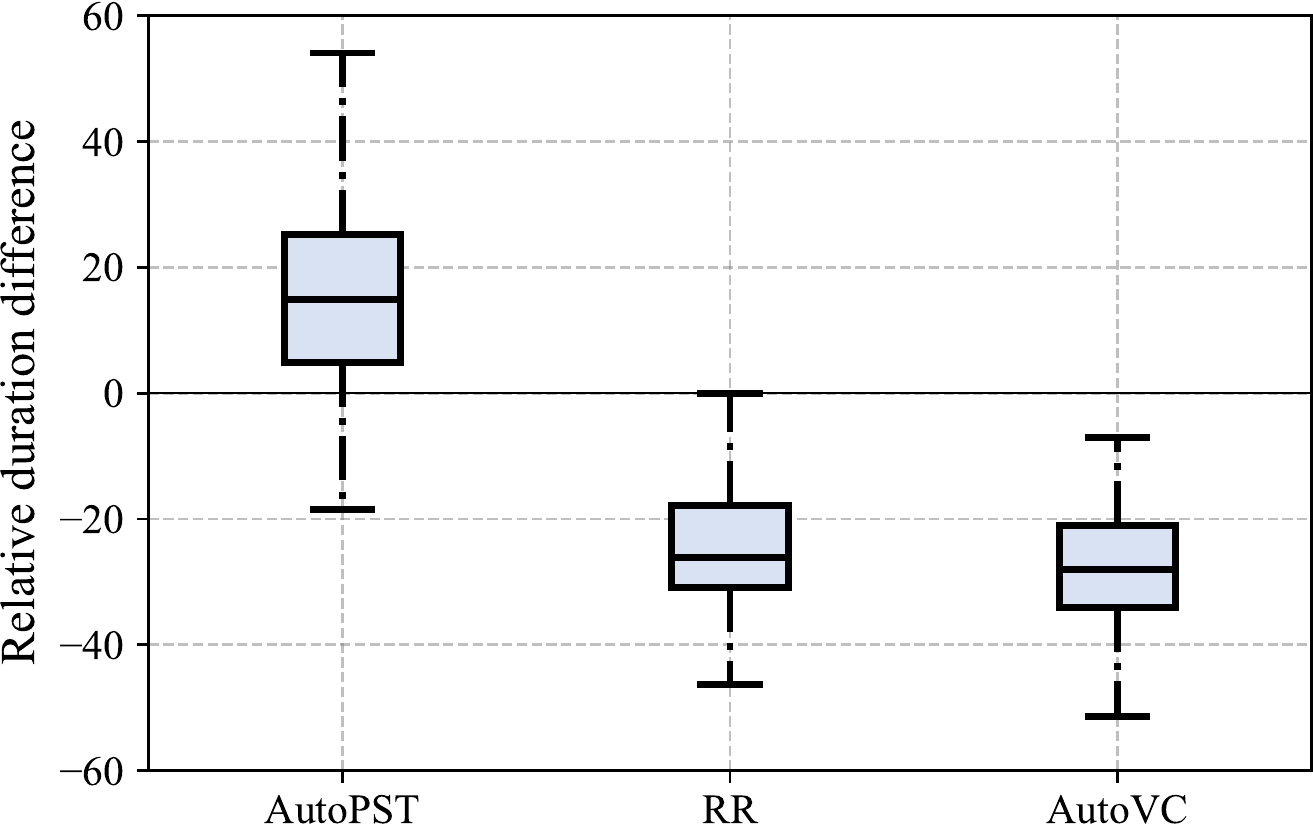}
        \caption{\small{VCTK (100\% vs 0\%)}}\label{fig:box_vctk}
    \end{subfigure}
    
    \begin{subfigure}{\columnwidth}
        \centering
        \includegraphics[width=0.84\linewidth]{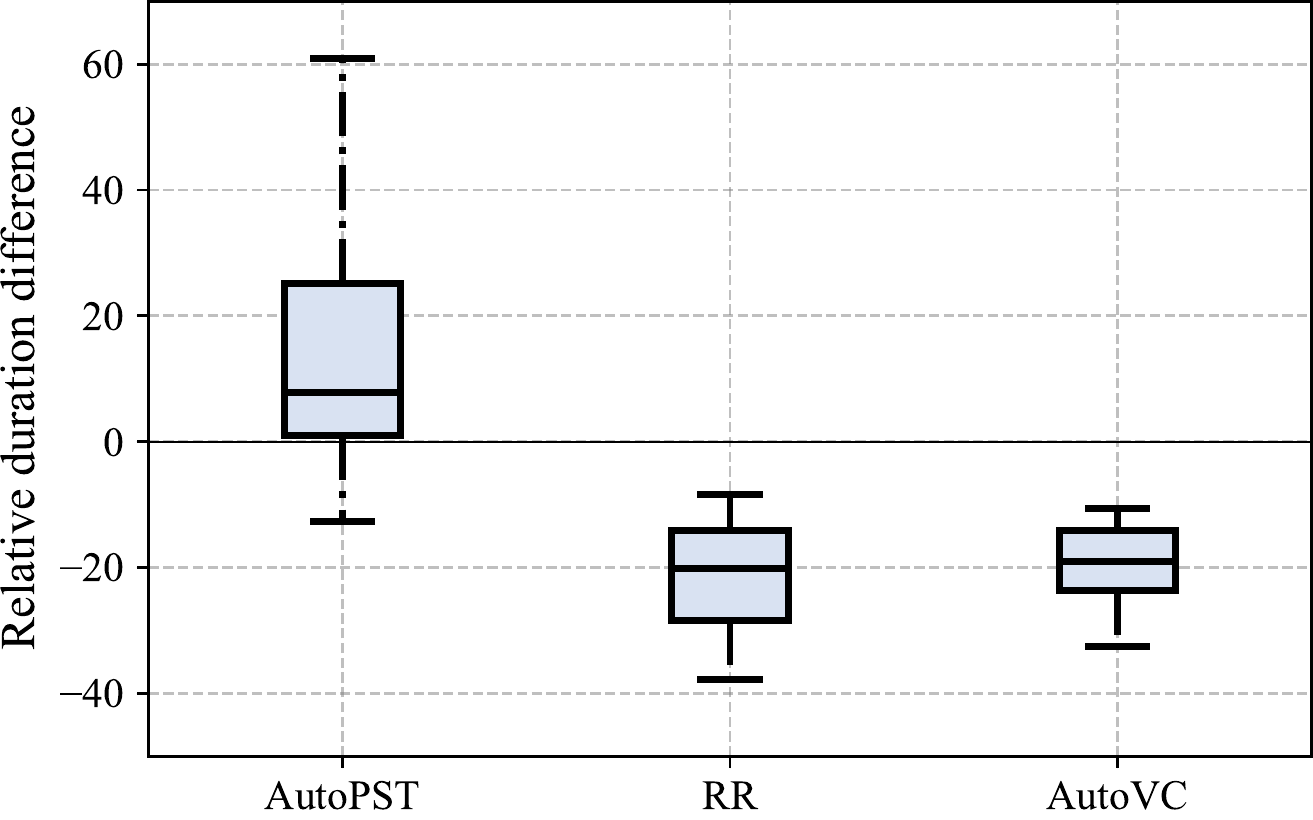}
        \caption{\small{VCTK (75\% vs 25\%)}}\label{fig:box_vctk_2575}
    \end{subfigure}
    
    \begin{subfigure}{\columnwidth}
        \centering
        \includegraphics[width=0.84\linewidth]{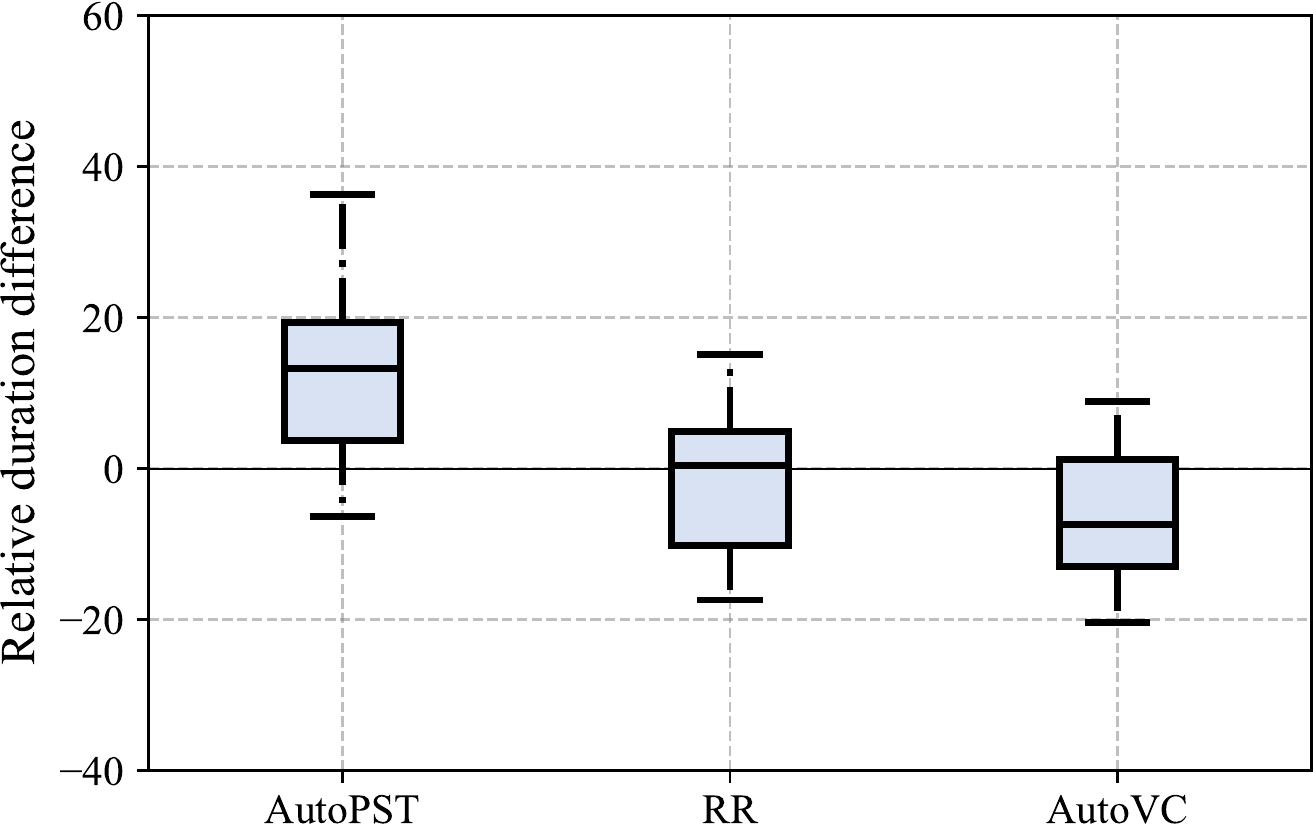}
        \caption{\small{VCTK (60\% vs 40\%)}}\label{fig:box_vctk_4060}
    \end{subfigure}
    
    \begin{subfigure}{\columnwidth}
        \centering
        \includegraphics[width=0.84\linewidth]{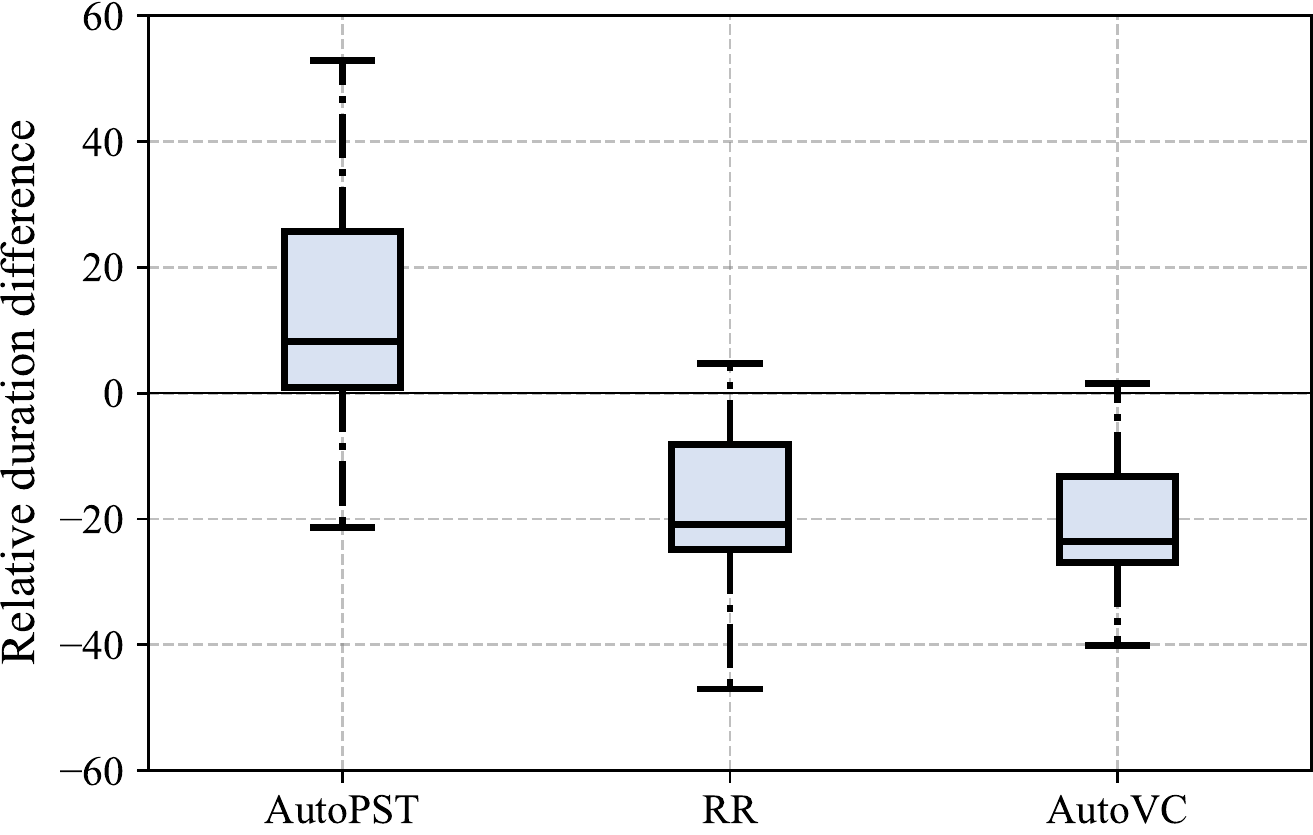}
        \caption{\small{Emo-VDB}}\label{fig:box_vdb}
    \end{subfigure}
\caption{The box plot of the relative duration difference between low-to-fast conversion and fast-to-slow conversion of utterance pairs with the same content. Positive duration differences indicate more sufficient rhythm disentanglement. Percentiles in the subcaptions denote rankings in speech rate from fast to slow.} 
\label{fig:box_plot_vctk}
\end{figure}

\subsection{Relative Duration Difference}
\label{subsec:dur_diff}

To objectively measure the extent to which each algorithm modifies rhythm to match the style of the target speaker, for each test sentence and each fast-slow speaker pair (the test set is a parallel dataset), 
we generate a fast-to-slow and a slow-to-fast conversion. If an algorithm does not change the rhythm at all, then the fast-to-slow version should have a shorter duration than its slow-to-fast counterpart; in contrast, if an algorithm can sufficiently obscure and ignore the rhythm information in the source speech, then it can flip the ordering. We compute relative duration difference as
\begin{equation}
    \small
    \mbox{Relative Duration Difference} = (L_\textrm{F2S} - L_\textrm{S2F}) / {L_\textrm{S2F}},
    \label{eq:dur_diff}
\end{equation}
where \e{L_\textrm{F2S}} and \e{L_\textrm{S2F}} denote the lengths of fast-to-slow and slow-to-fast conversions, respectively.
If the rhythm disentanglement is sufficient, this difference should be positive.

Figure~\ref{fig:box_plot_vctk}(\subref{fig:box_vctk}) shows the box plot of the relative duration differences across all test sentences and all the top four fastest-slowest speaker pairs. Figure~\ref{fig:box_plot_vctk}(\subref{fig:box_vctk_2575}) and Figure~\ref{fig:box_plot_vctk}(\subref{fig:box_vctk_4060}) show the results on speaker pairs with smaller rhythm differences (75\% vs 25\% and 60\% vs 40\% respectively). As shown, only \algname\ can achieve a positive average relative duration difference, which verifies its ability to disentangle rhythm. \textsc{AutoVC} gets the most negative relative duration differences, which is expected because it does not modify duration at all. RR achieves almost the same negative duration differences, which verifies that its rhythm disentanglement is insufficient. 

\begin{table}[]
    \centering
    \small
    \caption{Subjective evaluation results.}
    \begin{tabular}{l|ccc}
        \hline\hline
         & \algname & RR & \textsc{AutoVC} \\
        \hline
        \textbf{Timbre} & \textbf{4.29 $\pm$ 0.032} & 4.07 $\pm$ 0.037 & 4.26 $\pm$ 0.034\\
        \textbf{Prosody} & \textbf{3.61 $\pm$ 0.053} & 2.97 $\pm$ 0.063 & 2.64 $\pm$ 0.066\\
        \textbf{Overall} & \textbf{3.99 $\pm$ 0.036} & 3.63 $\pm$ 0.045 & 3.49 $\pm$ 0.052\\
        \hline\hline
    \end{tabular}
    \label{tab:subject}
\end{table}

\subsection{Subjective Evaluation}

To better evaluate the overall quality of prosody style transfer, and whether prosody style transfer improves the perceptual similarity to the target speaker, we performed a subjective evaluation on Amazon Mechanical Turk. 
Specifically, in each test unit, the subject first listens to two randomly ordered reference utterances from the source and target speaker respectively. Then the subject listens to a converted utterance from the source to the target speaker by one of the algorithms. Note that the content of the converted utterance is different from that in the reference utterances. Finally, the subject is asked to assign a score of 1-5 to describe the similarity to the target speaker in one of the three aspects: \emph{prosody similarity}, \emph{timbre similarity}, and \emph{overall similarity}. 
A score of 5 means entirely like the target speaker; 1 means completely like the source speaker; 3 means somewhat between the two speakers. Each algorithm has 79 utterances, where each utterance is assigned to 5 subjects.

Table~\ref{tab:subject} shows the subjective similarity scores. As shown, \algname\  has a significant advantage in terms of prosody similarity over the baselines, which further verifies that \algname\ can generate a prosody style that is perceived as similar to the target speaker. In terms of timbre similarity, \algname\  performs on-par with \textsc{AutoVC}, and the gaps among the three algorithms are small because all three algorithms apply the same mechanism to disentangle timbre. 

For the overall similarity, it is interesting to see how the subjects weigh the different aspects in their decisions. Specifically, although \textsc{AutoVC} can satisfactorily change the timbre, it is still perceived as only very slightly similar to the target speaker, because the prosody is not converted at all. In contrast, the \algname\  results, with all aspects transferred, are perceptually much more similar to the target speaker. This result shows that prosody indeed plays an important role in characterizing a speaker's style and should be adequately accounted for in speech style transfer.

\subsection{Restoring Abnormal Rhythm Patterns}

So far, our experiments have mostly focused on the overall speech rate. One question we are interested in is whether \algname\ can recognize fine-grained rhythm patterns, or can only adjust speech rate globally. To study this question, we modify the utterance \textit{``One showing mainly red and yellow''} in Figure~\ref{fig:spect_convert}(a) by stretching \textit{``yellow''} by two times, creating an abnormal rhythm pattern. We then let RR and \algname\  reconstruct the utterance from this abnormal input. If these algorithms can recognize fine-grained rhythm patterns, they should be able to restore the abnormality.

Figure~\ref{fig:restore} shows the reconstructed spectrogram from the abnormal input. As shown, RR attempts to reduce the overall duration, but it seems unable to reduce the abnormally long word \textit{``yellow''} more than the other words. In contrast, \algname\ not only restores the overall length, but also largely restores the word \textit{``yellow''} to its normal length. This shows that \algname\ can indeed capture the fine-grained rhythm patterns instead of blindly adjusting the speech rate.


\subsection{Emotion Style Transfer}
\label{subsec:emotion}

Although \algname\ is designed for voice style transfer, we nevertheless also test \algname\ on the much harder non-parallel emotion style transfer to investigate the generalizability of our proposed framework. We use the EmoV-DB dataset \cite{adigwe2018emotional}, which contains acted expressive speech of five emotion categories (amused, sad, neutral, angry, and sleepy) from four speakers. During training, two emotion categories are randomly chosen for each speaker and held out, for the purpose of evaluating generalization to unseen emotion categories for each speaker.

Among the five emotions, neutral has the fastest speech rate and sleepy has the slowest. We thus follow the same step in Section~\ref{subsec:dur_diff} to compute the relative duration difference between fast-to-slow and slow-to-fast emotion conversions for each speaker. Figure~\ref{fig:box_plot_vctk}(\subref{fig:box_vdb}) shows the box plot. Consistent with the observations in Section~\ref{subsec:dur_diff}, \algname\ can bring most of the relative duration differences to positive numbers, whereas the baselines cannot. This result shows that \algname can generalize to other domains. Additional results can be found in Appendix~\ref{append:exper_results}.

\section{Conclusion}
\begin{figure}[t]
\centering
\includegraphics[width=0.95\linewidth]{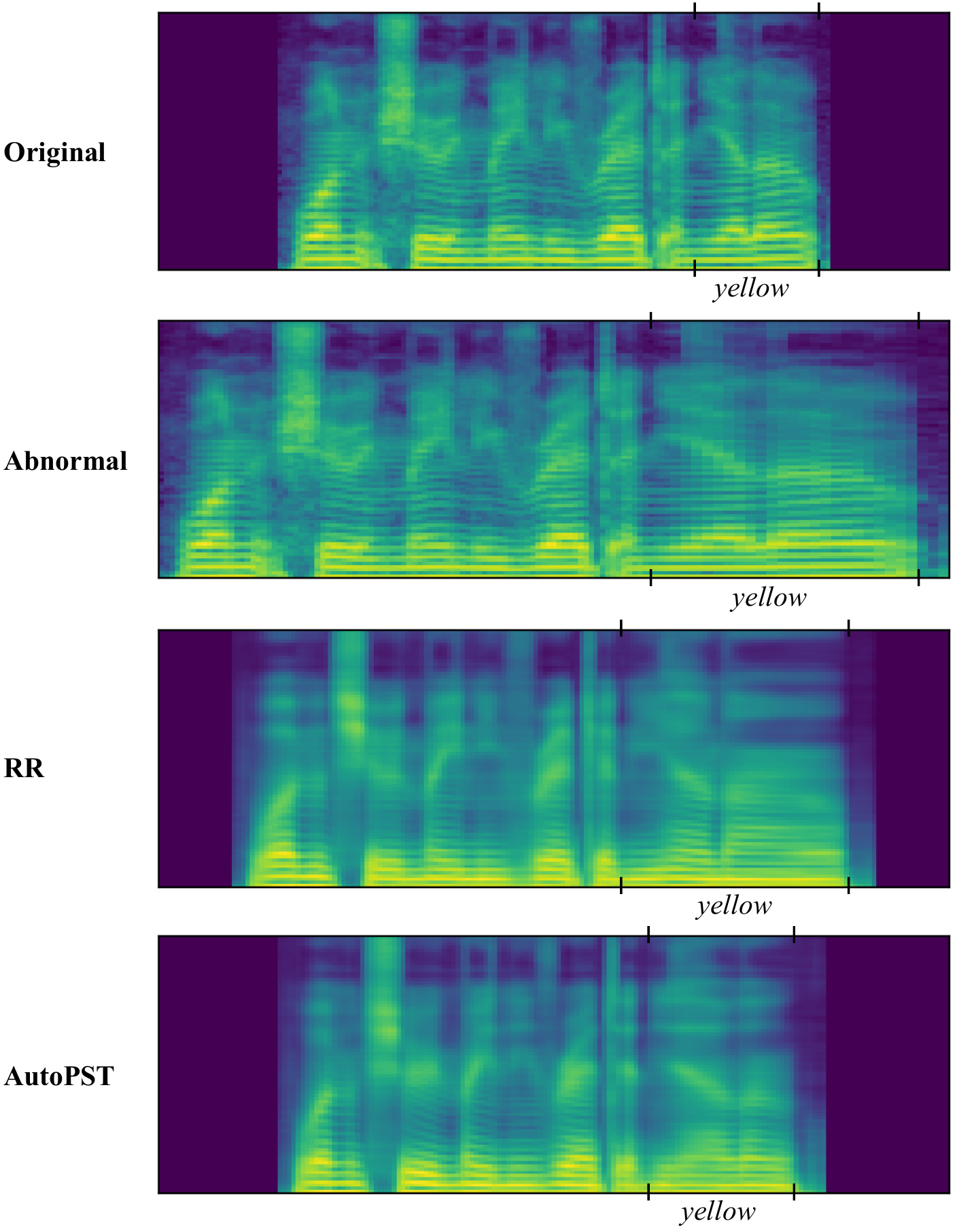}
\caption{Mel-spectrograms of reconstruction from \textit{``One showing mainly red and yellow''} with an abnormally long \textit{``yellow''}. \algname\ can largely restore the duration to its original length.}
\label{fig:restore}
\end{figure}

  In this paper, we propose \algname\ an autoencoder based prosody style transfer algorithm that does not rely on text transcriptions or fine-grained local prosody style information. We have empirically shown that \algname\ can effectively convert prosody, particularly the rhythm aspect, to match the target domain style. There are still some limitations of \algname. Currently, in order to disentangle the timbre information, \algname\ introduces a very harsh limitations on the dimension of the hidden representations. However, this would compromise the quality of the converted speech. How to strike a better balance between timbre disentanglement, prosody disentanglement, and audio quality remains a challenging future direction. Also, it has been shown by \citet{deng2021unsupervised} that \textsc{AutoVC} performs poorly when given in-the-wild audio examples. Since \algname\ inherits the basic framework of \textsc{AutoVC}, it is unlikely to generalize well to in-the-wild examples either. Improving the generalizability to audios recorded in different environments is a future direction. 
  
 \section*{Acknowledgment}

We would like to give special thanks to Gaoyuan Zhang from MIT-IBM Watson AI Lab, who has helped us a lot with building our demo webpage.


\bibliography{ref}
\bibliographystyle{icml2021}

\appendix
\newpage

\section{Rhythm Disentanglement from an Information-Theoretic Perspective}
\label{append:theory}
In this appendix, we will provide further explanation on why Equation~\eqref{eq:nes_cond} is necessary for rhythm disentanglement, \emph{i.e.} reducing \e{I(\bm R; \tilde{\bm Z})}. Before the formal analysis, we state the following assumption:
\begin{equation}
    \small
    H(\bm S | \bm X) = 0, \quad H(\bm R | \bm X) = 0,
    \label{eq:main_assump}
\end{equation}
which means that the phonetic symbol and the rhythm information is completely determined given the speech utterance. 
In the following, we will show two theorems, which serve as an elaboration of the brief discussion below Equation~\eqref{eq:nes_cond}.
\begin{theorem}
Assume that Equations~\eqref{eq:constraint} and \eqref{eq:main_assump} hold. If \e{\forall \bm X(t)} and \e{\bm X'(t)} with the same content information \e{\bm S}, but with different rhythm information, \e{\bm R} and \e{\bm R'} respectively,
\begin{equation}
\small
    Pr(\tilde{\bm Z} = \tilde{\bm Z}') = 0,
    \label{eq:thm1_assump}
\end{equation}
then
\begin{equation}
\small
    I(\bm R; \tilde{\bm Z}) = I(\bm R; \bm X).
\end{equation}
\end{theorem}
\begin{proof}
According to \eqref{eq:thm1_assump}, when \e{\bm S} given, there forms an injective mapping from \e{\bm R} to \e{\tilde{\bm Z}}, and thus
\begin{equation}
    \small
    H(\bm R | \tilde{\bm Z}, \bm S) = 0.
    \label{eq:H_R|Z_S}
\end{equation}
On the other hand, according to Equation~\eqref{eq:constraint}
\begin{equation}
    \small
    I(\bm S ; \tilde{\bm Z}) = I (\bm S; \bm X) = H(\bm S) - H(\bm S | \bm X) = H(\bm S),
\end{equation}
where the last equality is due to Equation~\eqref{eq:main_assump}, and thus
\begin{equation}
    \small
    H(\bm S | \tilde{\bm Z}) = I(\bm S ; \tilde{\bm Z}) - H(\bm S) = 0.
    \label{eq:H_S|Z}
\end{equation}
Plug Equation~\eqref{eq:H_S|Z} into Equation~\eqref{eq:H_R|Z_S}, we have
\begin{equation}
    \small
    H(\bm R | \tilde{\bm Z}) = 0.
    \label{eq:H_R|Z}
\end{equation}
Therefore
\begin{equation}
    \small
    \begin{aligned}
    I(\bm R; \tilde{\bm Z}) &= H(\bm R) - H(\bm R | \tilde{\bm Z}) \\
    &= H(\bm R) \\
    &= H(\bm R) - H(\bm R | \bm X) \\
    &= I(\bm R; \bm X),
    \end{aligned}
\end{equation}
where the second line is given by Equation~\eqref{eq:H_R|Z}; the third line is given by Equation~\eqref{eq:main_assump}.
\end{proof}
\begin{theorem}
Assume that Equations~\eqref{eq:constraint} and \eqref{eq:main_assump} hold. If \e{\forall \bm X(t)} and \e{\bm X'(t)} with the same content information \e{\bm S}, but with different rhythm information, \e{\bm R} and \e{\bm R'} respectively,
\begin{equation}
\small
    Pr(\tilde{\bm Z} = \tilde{\bm Z}') = 1,
    \label{eq:thm2_assump}
\end{equation}
then
\begin{equation}
\small
    I(\bm R; \tilde{\bm Z}) = I(\bm R; \bm S).
\end{equation}
\end{theorem}
\begin{proof}
According to \eqref{eq:thm2_assump}, when \e{\bm S} given, \e{\tilde{\bm Z}} is a constant regardless of the value of \e{\bm R}, and thus
\begin{equation}
\small
    H(\bm R | \tilde{\bm Z}, \bm S) = H(\bm R | \bm S).
    \label{eq:H_R|Z_S_2}
\end{equation}
Plug Equation~\eqref{eq:H_S|Z} into Equation~\eqref{eq:H_R|Z_S_2}, we have
\begin{equation}
\small
    H(\bm R | \tilde{\bm Z}) = H(\bm R | \bm S).
\end{equation}
Therefore
\begin{equation}
    \small
    \begin{aligned}
        I(\bm R; \tilde{\bm Z}) &= H(\bm R) - H(\bm R | \tilde{\bm Z}) \\
        &= H(\bm R) - H(\bm R | \bm S) \\
        &= I(\bm R; \bm S).
    \end{aligned}
\end{equation}
\end{proof}

\section{Additional Algorithm Details}
\label{append:sea}

In this appendix, we will cover the additional algorithm details of \algname. 
Our implementation is available very soon.

\subsection{Self-Expressive Autoencoder (SEA)}
\label{subsec:sea}

SEA aims to derive an representation \e{\bm A(t)}, which is very similar among similar frames, and very disimilar among disimilar frames. SEA consists of one encoder, which derives \e{\bm A(t)} from input speech, and two decoders. One decoder reconstructs the input based on \e{\bm A(t)}. The other decoder reconstructs the input based on another representation \e{\bm B(t)}, which is computed as follows
\begin{equation}
    \small
    \bm B(t) = \sum_{t' \neq t} \left(\frac{\bm A^T(t') \bm A(t)}{\Vert \bm A(t')\Vert_2 \Vert\bm A(t)\Vert_2} \right) \cdot \bm A(t').
    \label{eq:sea}
\end{equation}
The entire pipeline is trained jointly to minimize the reconstruction loss of both decoders. The intuition behind this self-expressive mechanism is that in order to achieve good reconstructions at both decoders, \e{\bm A(t)} should be similar to \e{\bm B(t)}. According to Equation~\eqref{eq:sea}, \e{\bm A'(t)} is essentially a linear combination of the representation of all the other frames, and the combination weight is the cosine similarity. If frame \e{t'} is dissimilar to the current frame \e{t} to be expressed, the weight has to be close to zero, otherwise \e{\bm B(t)} would be made dissimilar to \e{\bm A(t)}; if frame \e{t'} is similar to the current frame \e{t} to be expressed, the weight has to be close to one to ensure it is contributing sufficiently to \e{\bm B(t)}.

\subsection{F0 and UV Conditioning}

One of the biggest challenges of \algname\  is to infer the F0 contour and voiced/unvoiced (UV) states of the utterances based solely on the input MFCC features, because it is very hard for the system to elicit lexical, semantic and syntactic information without text transcriptions. To mitigate this problem, we introduce two optional conditioning, one on the input F0 contour and one on UV states. The sequence being conditioned upon is concatenated with the encoder output along the channel dimension \emph{before} being sent to the similarity-based resampling module. The concatenation is valid because the conditioning sequence has the same temporal length as the encoder output, which are both equal to the temporal length of the input.

Although the F0 conditioning will largely resolve the ambiguity of reconstructing the F0 contour, it will make the algorithm unable to convert the F0 aspect of prosody, which is undesirable in many applications. On the other hand, UV conditioning can still resolve the ambiguity of reconstructing the F0 information to some extent, while maintaining \algname's ability to disentangle the pitch aspect. In practice, we can choose F0 conditioning, UV conditioning, or neither depending on different applications.

\subsection{Domain Identity Conditioning}

According to Equation~\eqref{eq:framework}, the decoder takes the domain identity, \e{D}, as a second input. This is achieved by first feeding \e{D} to a feedforward layer. Then, the feedforward layer is appended to the first time step of the encoder output sequence, \e{\tilde{\bm Z}(t)}, as well as to the first time step of the memory of the decoder (Our decoder is a Transformer; the memory is the output of the Transformer encoder). In other words, the total length of the encoder output and the memory will increase by one after the appending, with the first time step being the appended domain identity. Since the decoder is a Transformer, it can elicit the domain identity information by attending to the first time step.

\subsection{Output and Losses}

In addition to outputting the reconstructed spectrogram, \e{\hat{\bm S}(t)}, the \algname\ decoder also outputs a stop token prediction. Stop token, denoted as \e{P(t)}, guides when the sequential spectrogram generation should stop. It is a scalar sequence that equals zero at time steps within the lengths of the ground truth spectrogram, and equals one at time steps after the ground truth spectrogram has ended. Denote the predicted stop token as \e{\hat{P}(t)}. Then the total loss function consists of the \e{\ell_2} loss for spectrogram prediction and the cross-entropy loss for stop token prediction:
\begin{equation}
    \small
    \begin{aligned}
    \mathcal{L} = &\mathbb{E}_{train}\Big[\sum_{t=1}^T \Vert \hat{\bm S}(t) - \bm S(t)\Vert_2^2 + \\
    & \sum_{t=1}^{T+k} w P(t) \log \hat{P}(t) + (1-P(t)) \log (1-\hat{P}(t))\Big].
    \end{aligned}
\end{equation}
The expectation is taken over the training set. $T$ denotes the total length of the spectrogram. Notice that in the second summation, $t$ runs to from 1 to \e{T+k}, which $k$ time steps longer than the spectrogram. This is because the ground truth stop token is always zero for \e{t \leq T}. It is only equal to one for \e{t > T}. So we set $k$ steps for the positive samples.

Nevertheless, the negative examples still occur much more often than the positive examples. To fix the unbalanced label problem, we add a positive weight, \e{w}, to the positive class.


\section{Experiment Details}
\label{append:exper_detail}

In this appendix, we will cover the additional details of our experiments.

\subsection{Architecture}

The \algname\ encoder is a simple 8-layer 1D convolutional network, where each layer uses $5 \times 1$ filters with $1$ stride, SAME padding and ReLU activation. GroupNorm \cite{wu2018group} is applied to every layer. The number of filters is 512 for the first five layers, and the last three layers have 128, 32, and 4 filters respectively. The \algname\ decoder is a Transformer \cite{vaswani2017attention}, which has four encoder layers and four decoder layers. The model dimension is 256 and the number of heads is eight. 

\subsection{Training}

We implement our model using Pytorch 1.6.0, and we train our model on a single Tesla V100 GPU using Adam optimizer. Synchronous training takes $3\times 10^5$ steps with a batch size of 4. Asynchronous training takes $6\times 10^5$ steps with a batch size of 4. We use Pre-LN \cite{xiong2020layer} without the warm-up stage.

\subsection{Datasets}

For the VCTK dataset, our test set consists of parallel utterances from 24 speakers. In order to identify speakers with the fastest and slowest speech rates, we take the average of the log duration of the test utterances. Since all the utterances are spoken in parallel by all of the speakers, the speaker with the smallest log duration can be considered as the fastest speaker, and the speaker with the largest log duration can be considered as the slowest speaker. We use the log duration instead of duration because the differences between the average log duration can be nicely interpreted as the average percentage difference in duration of the same sentence uttered by the two speakers. We then select the two fastest speakers (P231 and P239), and the two slowest speakers (P270 and P245) for our main evaluation. For further evaluation, we select two speaker pairs with smaller rhythm differences, one with speakers ranking 25\% and 75\% (P244 and P226) in speech rate; the other with ranking 40\% and 60\% (P240 and P256) in speech rate. For the Emo-VDB dataset, we follow a similar protocol of finding the fastest emotion, which is neutral, and the slowest emotion, which is sleepy, for our evaluation.

\subsection{Subjective Evaluation}

For the subjective evaluation, 18 sentences are generated for each fast-slow speaker pair (9 for fast-to-slow conversions and 9 for slow-to-fast conversions), and there are four fast-slow speaker pairs, summing to 72 utterances for each algorithm. Each utterance is assigned to five subjects. When evaluating voice similarity, the subjects are explicitly asked to focus only on voice but not on prosody; when evaluating prosody similarity, the subjects are asked to focus only on prosody but not on voice; when evaluating the overall speaker similarity, the subjects are asked to pay attention to all the aspects of speech, including voice and prosody.

For each test, the reference utterances (one from the source speaker and one from the target speaker) are randomized and named speaker 1 and speaker 2. The reference utterances are different from the test utterance in terms of content. The subjects are asked to assign a score of 1-5 on whether the aforementioned aspects sound more similar to speaker 1 or speaker 2, with 1 meaning completely like speaker 1 and 5 meaning completely like speaker 2. These scores are then converted to the similarity between the source and target speakers, with 1 meaning completely like the source speaker and 5 meaning completely like the target speaker.

\section{Additional Experiment Results}
\label{append:exper_results}

In this section, we will present some additional visualization and ablation study results.

\subsection{Similarity-based Resampling Visualization}
\label{subsec:vis}

In order to intuitive show the effect of our similarity-based resampling module, we design the following experiment. We first train a variant of \algname\ \emph{without} the random resampling module, so that it can generate a spectrogram that is synchronous with the hidden representations. Then, we performed the similarity-based random resampling on the hidden representation, and generate a \emph{time-synchronous} spectrogram from the resampled representations. In this way, we can intuitively see how much each segment is being shortened/lengthened by observing the time-synchronous spectrogram. 

Figure~\ref{fig:resample_vis} shows the reconstructed spectrograms from the resampled code of the utterance \emph{``Please call Stella''}. The left figures corresponds to the downsampling case, with \e{\tau} dropping from \e{0.98} down to \e{0.9}. The right figures corresponds to the upsampling case, with \e{\tau} increasing from \e{1.02} up to \e{1.1}. The top figure (Figure~\ref{fig:resample_vis}(\subref{subfig:1.0})) is the reference spectrogram without resampling. There are two observations. First, the total length of the code decreases as \e{\tau} decreases, and increases as \e{\tau} increases. Second, the relatively steady segments get stretched or shortened most, such as the \emph{``a''} segment in the second word.

\subsection{Upsampling v.s. Downsampling}

\algname\ adopts both similarity-based upsampling (Section~\ref{subsec:downsample}) and similarity-based downsampling (Section~\ref{subsec:upsample}). This section shows why both are necessary. In particular, we trained two variants of \algname, one without upsampling, and one without downsampling. All the other settings remain the same. We then perform fast-to-slow and slow-to-fast conversions the same way as in Section~\ref{subsec:dur_diff}. For each conversion, we compute the relative duration difference with respect to the \emph{original source speech}, \emph{i.e.}
\begin{equation}
    \small
    \begin{aligned}
    &\mbox{F2S Relative Duration Diff} = (L_{F2S} - L_{source}) / L_{source} \\
    &\mbox{S2F Relative Duration Diff} = (L_{S2F} - L_{source}) / L_{source}.
    \end{aligned}
\end{equation}
Note that this is different from the relative duration difference computed in Equation~\eqref{eq:dur_diff}. If an algorithm truly converts rhythm to the correct direction, it should have a positive F2S relative duration difference and a negative S2F relative duration difference.

Table~\ref{tab:up_vs_down} shows the relative duration differences. As can be seen, \algname\ can correctly change the rhythm to the desired direction. However, without either of the resampling module, the rhythm conversion becomes incorrect. If upsampling is removed, both fast-to-slow and slow-to-fast will increase the duration. If downsampling is removed, both fast-to-slow and slow-to-fast will decrease the duration. One possible explanation for this is that random resampling is only enforced during training. During testing, the random resampling will be removed (equivalent to the \e{\tau=1} case). If either downsampling or upsampling is removed, the test case, \e{\tau=1}, becomes a corner case, undesirably passing a duration bias. By having both upsampling and downsampling, we can ensure \e{\tau=1} is a well-represented mode among the training instances.

\begin{table}[t]
    \centering
    \caption{The relative duration difference with respect the original utterance. F2S denotes fast-to-slow conversion; S2F denotes slow-to-fast conversion. ``No Up'' denotes \algname\ with upsampling removed; ``No Down'' denotes \algname\ with downsampling removed. The numbers outside the parentheses represent the average; the numbers in the parentheses represent the standard deviation. The desired outcome should be a positive average for fast-to-slow and a negative average for slow-to-fast.}
    \vspace{0.1in}
    \begin{tabular}{c|ccc}
    \hline\hline
            & \algname & No Up & No Down  \\
    \hline
        F2S & 26.53 (23.05) & 32.79 (12.68) & -13.65 (6.46) \\
        S2F & -16.95 (23.70) & 13.47 (21.96) & -27.46 (5.00) \\
        \hline\hline
    \end{tabular}
    \label{tab:up_vs_down}
\end{table}

\subsection{Removing Two-Stage Training}

\begin{figure}[t]
\centering
\includegraphics[width=0.9\linewidth]{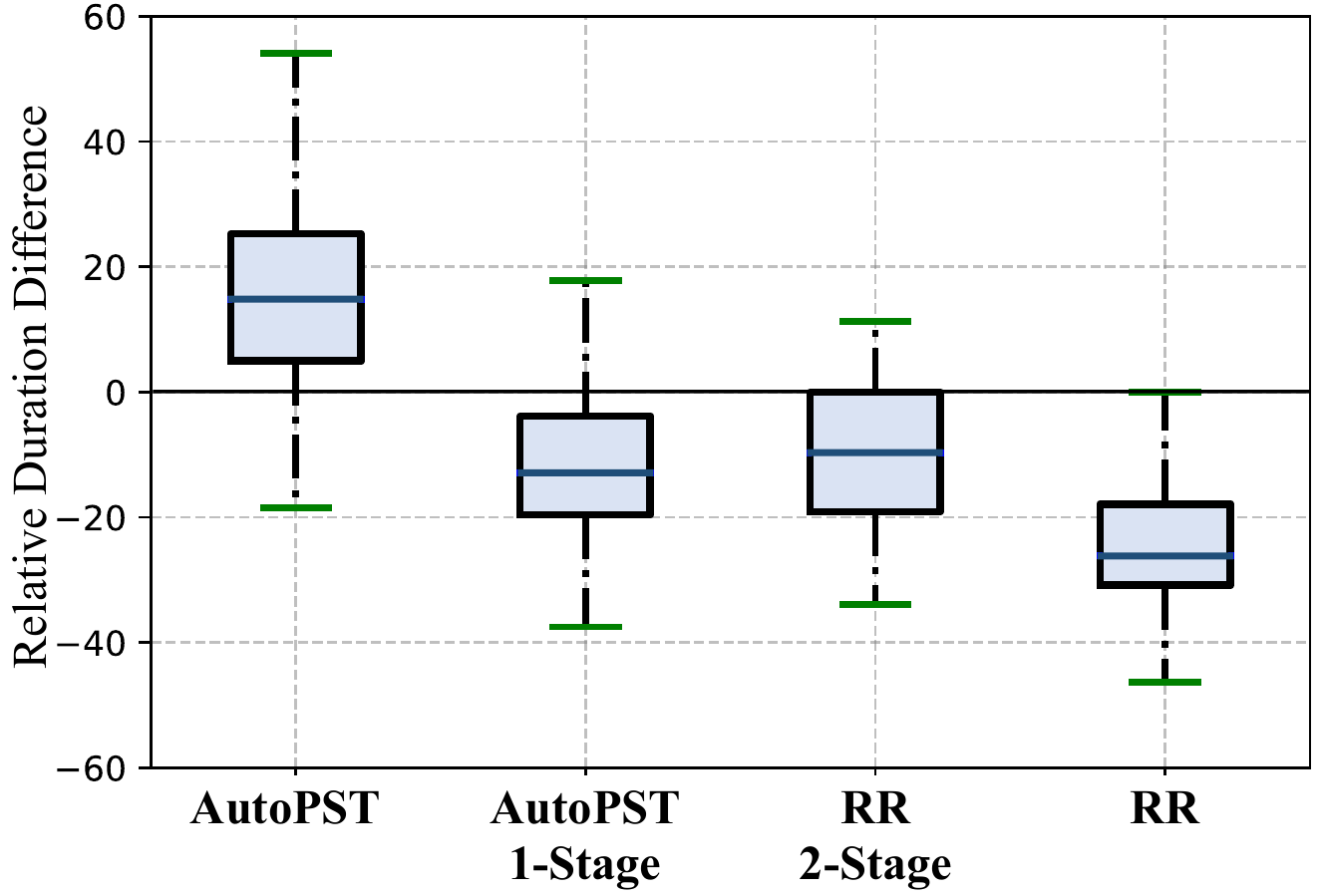}
\caption{Ablation study on two-stage training using the same relative duration difference experiment shown in Figure~\ref{fig:box_plot_vctk}. \algname\ 1-Stage denotes the \algname\ algorithm trained in an end-to-end manner (without two-stage training). RR 2-Stage denotes the RR baseline with the two-stage training.}
\label{fig:box_plot_ablation}
\end{figure}

As discussed, there are two mechanisms that promote prosody disentanglement. The first is similarity-based resampling; the second is two-stage training. In this section, we will explore how much each mechanism contributes to the performance advantage of \algname. In particular, we implement two variants of the algorithms. The first variant, called \algname\  1-Stage, removes the two-stage training of \algname, while all the other settings remain the same as \algname. The second variant, called RR 2-Stage, supplements the RR baseline with two-stage training. We then create the same box plot of relative duration difference as discussed Section~\ref{subsec:dur_diff}.

Figure~\ref{fig:box_plot_ablation} shows the results. As can be seen, without either two-stage training or similarity-based random resampling, the performance drops significantly, which implies that both mechanisms are essential for a successful rhythm disentanglement.

\subsection{Generalization to Unseen Emotions}

\begin{figure}[t]
\centering
\includegraphics[width=0.9\linewidth]{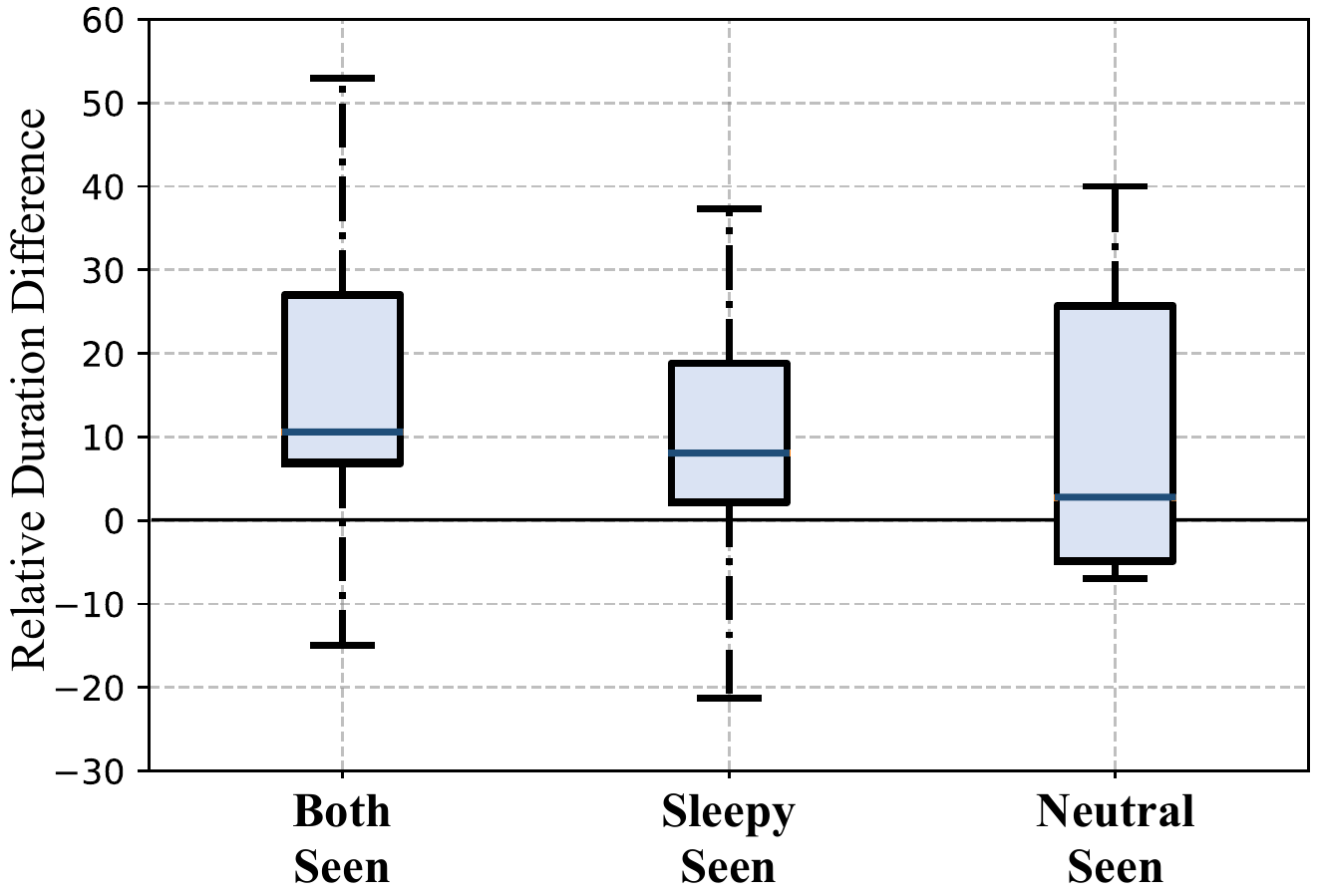}
\caption{Breakdown of Figure~\ref{fig:box_plot_vctk}(\subref{fig:box_vdb}) into 1) the speaker who has training examples of both neutral and sleepy, 2) the speakers who have training examples of sleepy only; and 3) the speaker who has training examples of neutral only.}
\label{fig:box_plot_emo_ablation}
\end{figure}

As mentioned in Section~\ref{subsec:emotion}, for the emotion conversion experiment, we deliberately remove certain emotion categories for each speaker from the training set. As a result, some speakers do not have training examples of either neutral or sleepy, or both. To examine whether \algname\  generalize to unseen speakers, we break the \algname\ samples in Figure~\ref{fig:box_plot_vctk}(\subref{fig:box_vdb}) into three groups. The first group consists of the speaker who has training examples of both neutral and sleepy. The second group consists of speakers who have training examples of only sleepy. The third group consists of the speaker who has training examples of only neutral.

Figure~\ref{fig:box_plot_emo_ablation} shows the box plot of the relative duration difference (same as Figure~\ref{fig:box_plot_vctk}(\subref{fig:box_vdb})) for these three groups. As can be seen, there is a slight performance advantage if both emotion categories are seen. However, even if there is one unseen emotion, the performance is still pretty competitive, demonstrating good generalizability to unseen emotions.

\subsection{Training Domain Embedding}
The domain ID is to assumed to present during testing, but we would like to explore the possibility of doing zero-shot conversion. Thus, we trained a variant of \algname\, where a domain encoder replaces the one-hot domain embedding. During testing, we only need to feed the domain encoder with a target speaker's utterance without needing the domain ID. Figure~\ref{fig:autopst_ge2e} shows the relative duration difference of this variant, which performs slightly worse than the original \algname\ due to the increased difficulty, but still significantly better than the baselines. 

\begin{figure}[t]
\centering
\includegraphics[width=0.9\linewidth]{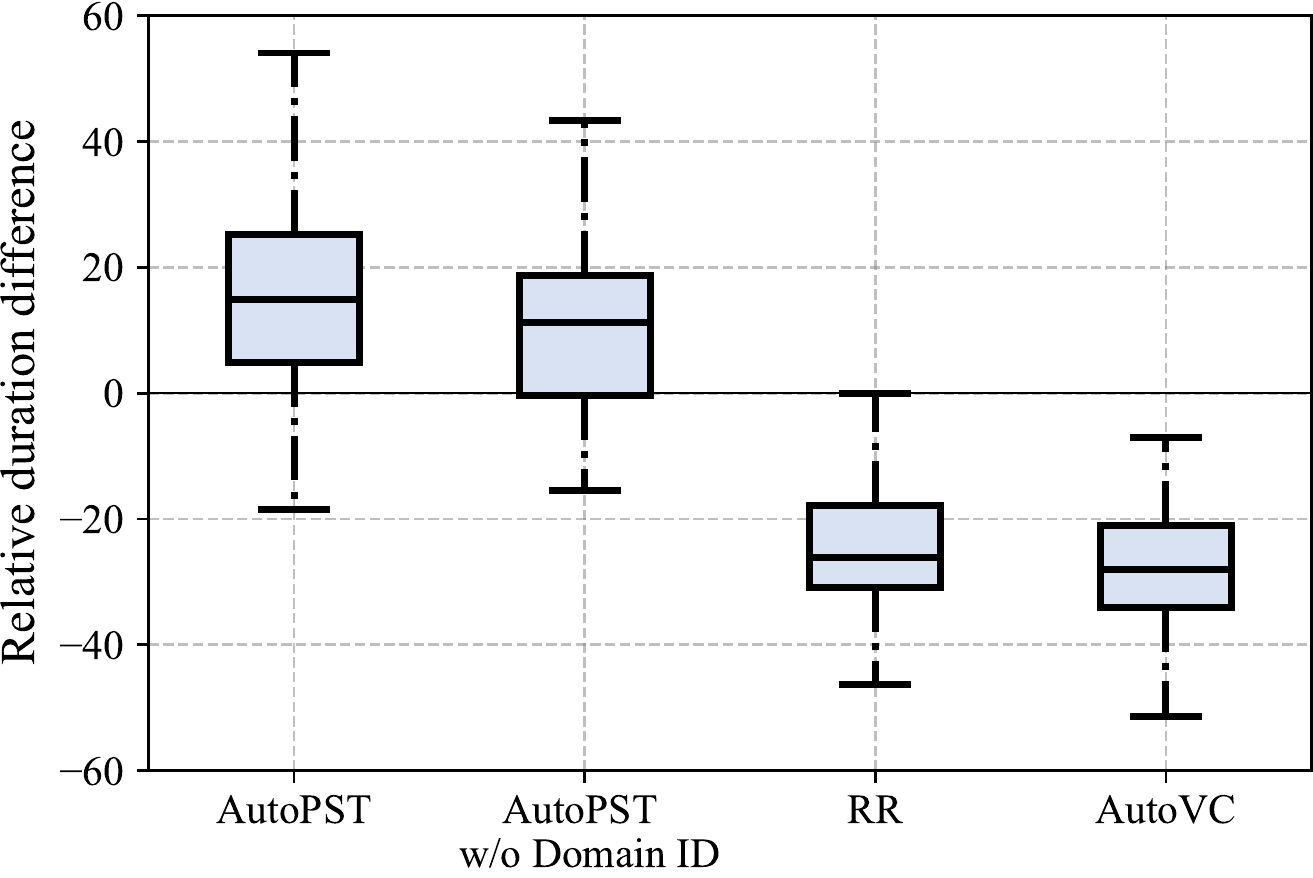}
\caption{The box plot of the relative duration difference between slow-to-fast conversion and fast-to-slow conversion of the same utterance pairs as used in Figure~\ref{fig:box_plot_vctk}(\subref{fig:box_vctk}). \algname\ w/o Domain ID denotes \algname\ variant trained using domain embedding.}
\label{fig:autopst_ge2e}
\end{figure}

\subsection{\textsc{SpeechSplit} Baseline}
Although \textsc{SpeechSplit} can also perform prosody style transfer, it requires ground truth target rhythm, \emph{i.e.} the target speaker speaking the source utterance. On the other hand, \algname\ seeks to perform prosody style transfer without the ground truth, which is a much harder task. Nevertheless, we show the relative duration difference of \textsc{SpeechSplit} in Figure~\ref{fig:spsp_baseline}. Note since the converted rhythm is very close to the ground truth, \textsc{SpeechSplit} is the performance upper-bound of any prosody conversion, including \algname. However, when there is no ground truth available, \textsc{SpeechSplit} becomes vulnerable. To show this, Figure~\ref{fig:spsp_baseline} also shows the result of \textsc{SpeechSplit} with a random target speaker utterance, instead of the ground truth target utterance, fed into its rhythm encoder. As can be observed, the relative duration difference almost completely concentrate around zero, which indicates that \textsc{SpeechSplit} completely fails in this case.

\begin{figure}[t]
\centering
\includegraphics[width=0.9\linewidth]{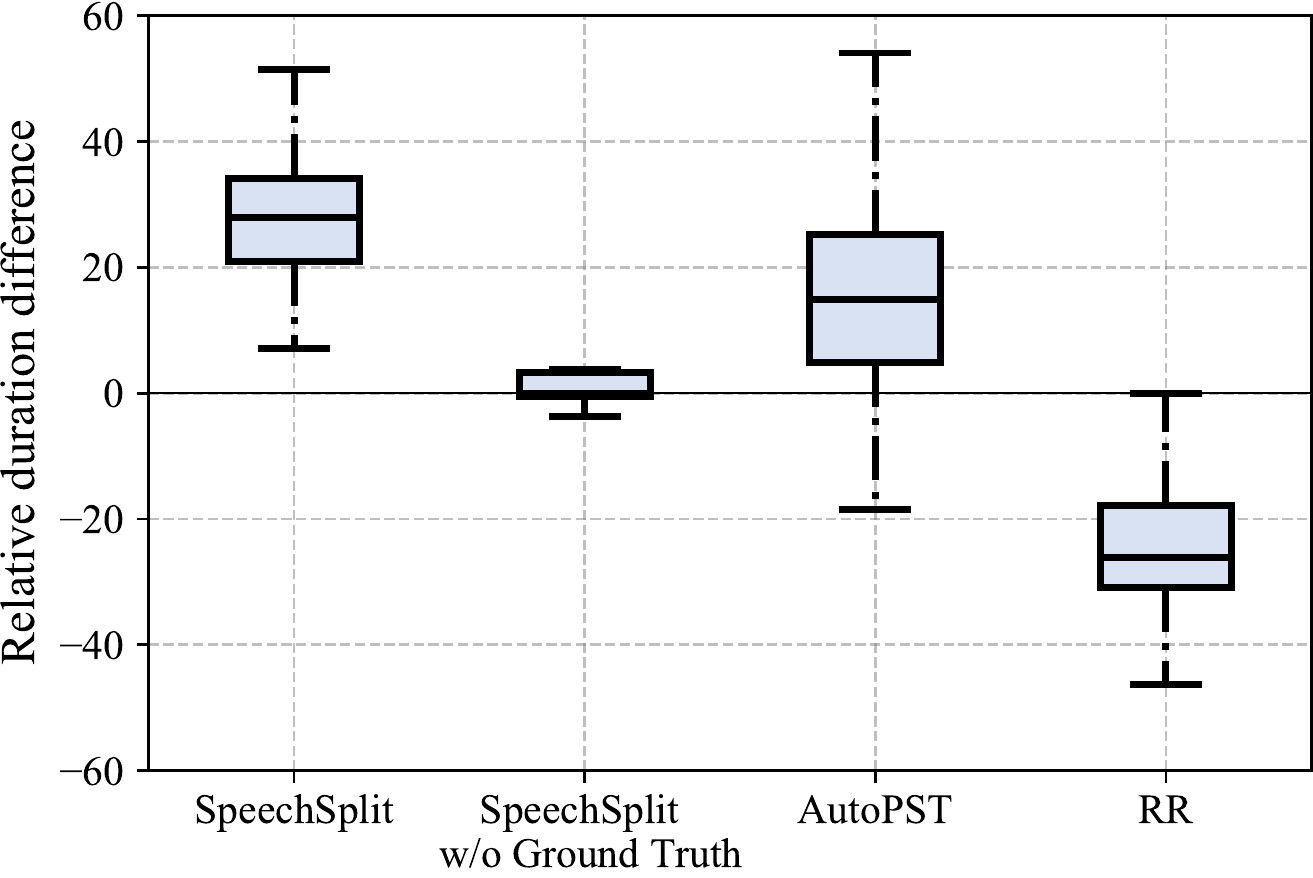}
\caption{The box plot of the relative duration difference between slow-to-fast conversion and fast-to-slow conversion of the same utterance pairs as used in Figure~\ref{fig:box_plot_vctk}(\subref{fig:box_vctk}). SpeechSplit w/o }
\label{fig:spsp_baseline}
\end{figure}




\begin{figure*}[t]
	\centering
	\begin{subfigure}{\columnwidth}
		\centering
		\includegraphics[width=0.9\linewidth]{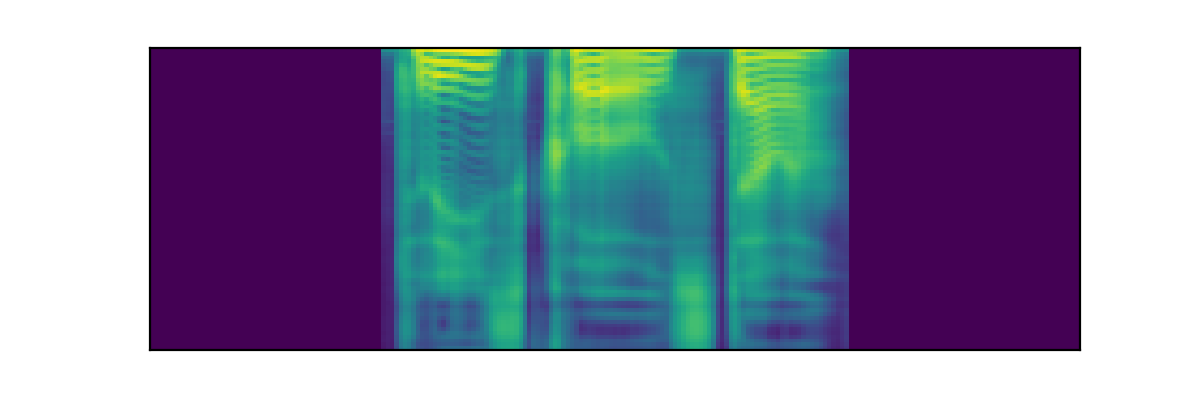}
		\caption{$\tau=1.0$}\label{subfig:1.0}
	\end{subfigure}
	\\
	\begin{subfigure}{\columnwidth}
		\centering
		\includegraphics[width=0.9\linewidth]{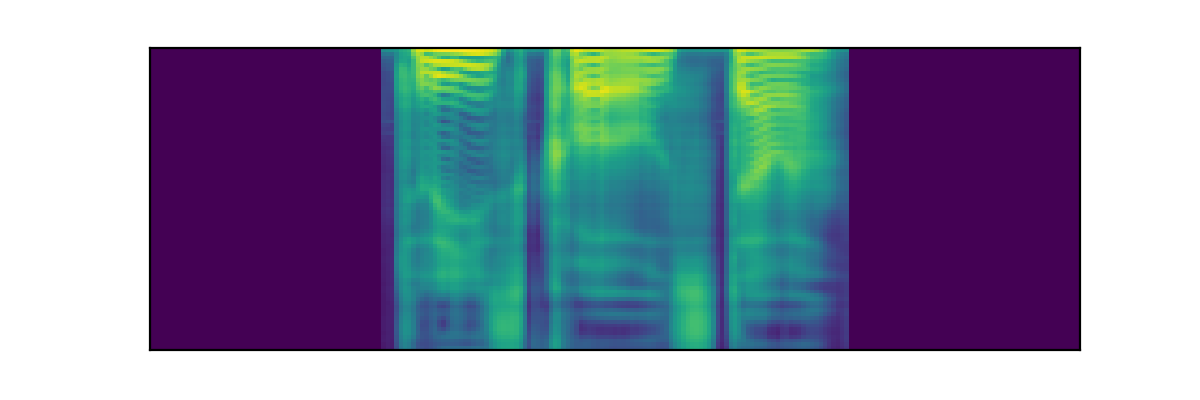}
		\caption{$\tau=0.98$}\label{subfig:0.98}
	\end{subfigure}
	\begin{subfigure}{\columnwidth}
		\centering
		\includegraphics[width=0.9\linewidth]{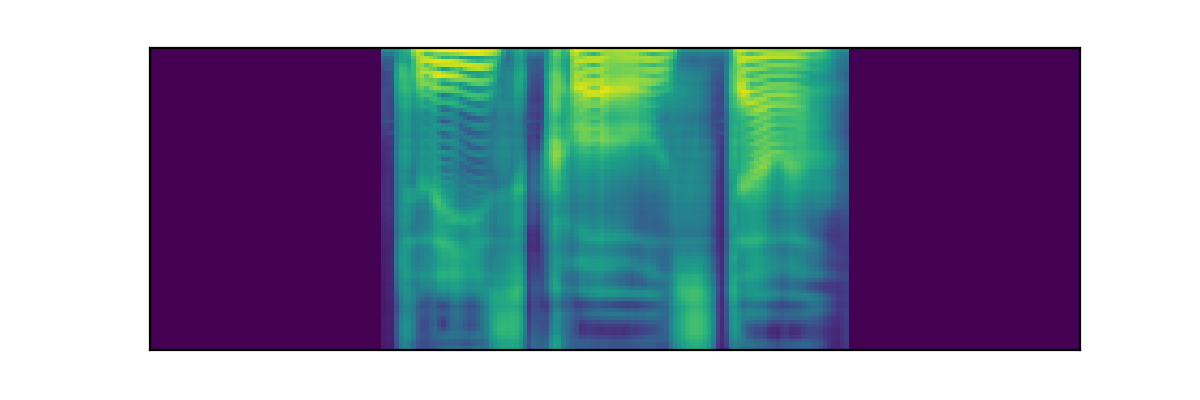}
		\caption{$\tau=1.02$}\label{subfig:1.02}
	\end{subfigure}
	\begin{subfigure}{\columnwidth}
		\centering
		\includegraphics[width=0.9\linewidth]{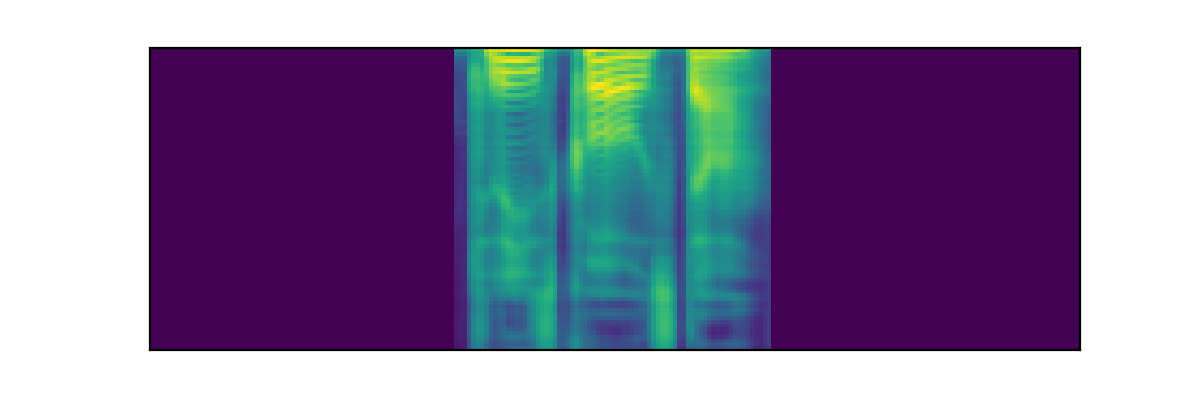}
		\caption{$\tau=0.96$}\label{subfig:0.96}
	\end{subfigure}
	\begin{subfigure}{\columnwidth}
		\centering
		\includegraphics[width=0.9\linewidth]{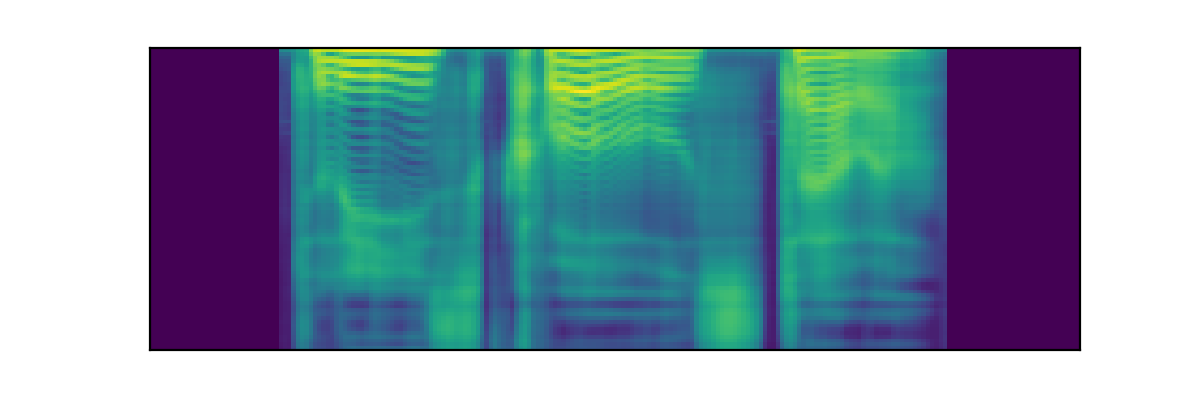}
		\caption{$\tau=1.04$}\label{subfig:1.04}
	\end{subfigure}
	\begin{subfigure}{\columnwidth}
		\centering
		\includegraphics[width=0.9\linewidth]{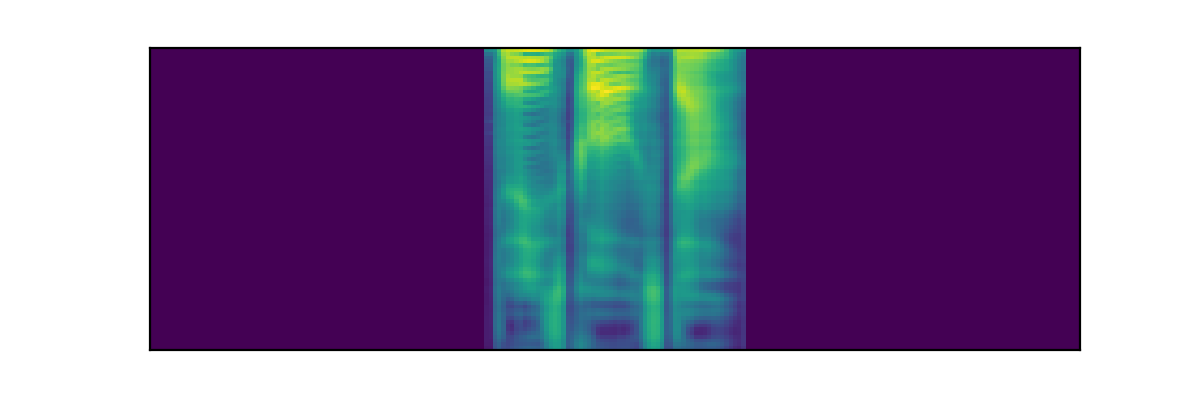}
		\caption{$\tau=0.94$}\label{subfig:0.94}
	\end{subfigure}
	\begin{subfigure}{\columnwidth}
		\centering
		\includegraphics[width=0.9\linewidth]{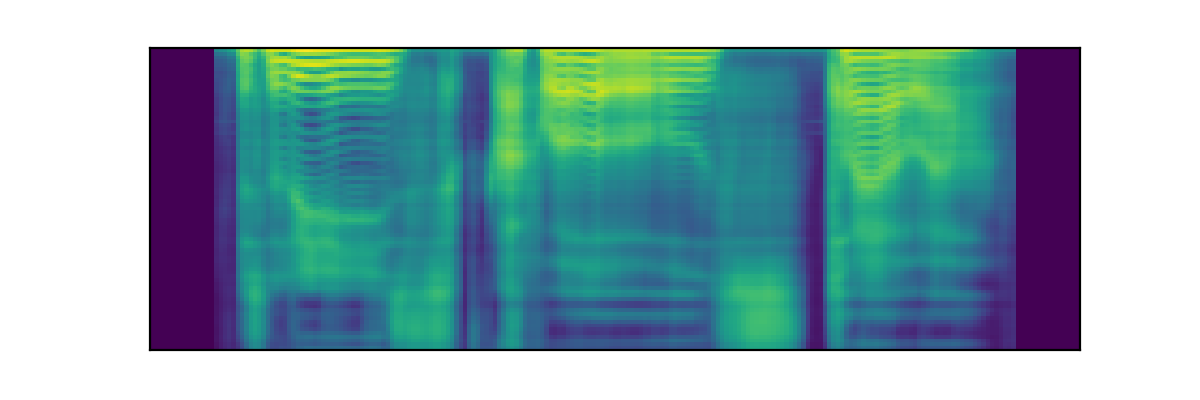}
		\caption{$\tau=1.06$}\label{subfig:1.06}
	\end{subfigure}
	\begin{subfigure}{\columnwidth}
		\centering
		\includegraphics[width=0.9\linewidth]{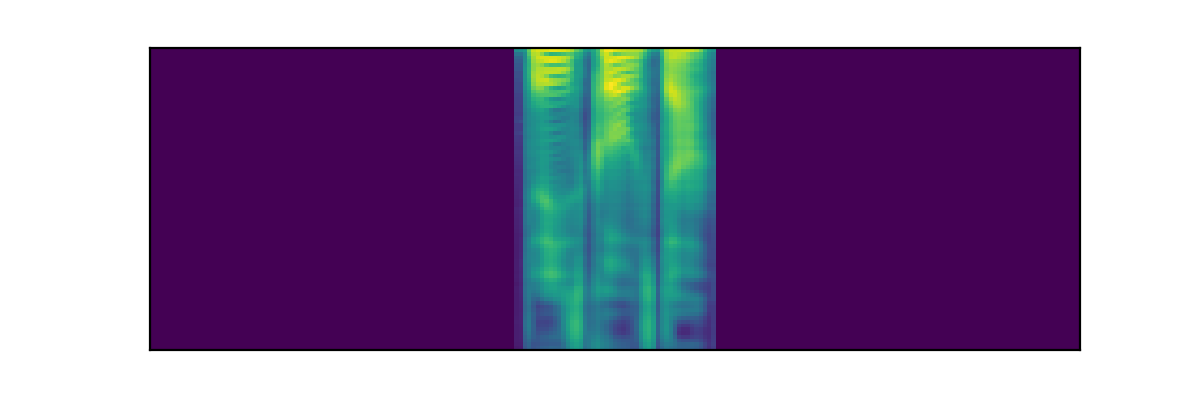}
		\caption{$\tau=0.92$}\label{subfig:0.92}
	\end{subfigure}
	\begin{subfigure}{\columnwidth}
		\centering
		\includegraphics[width=0.9\linewidth]{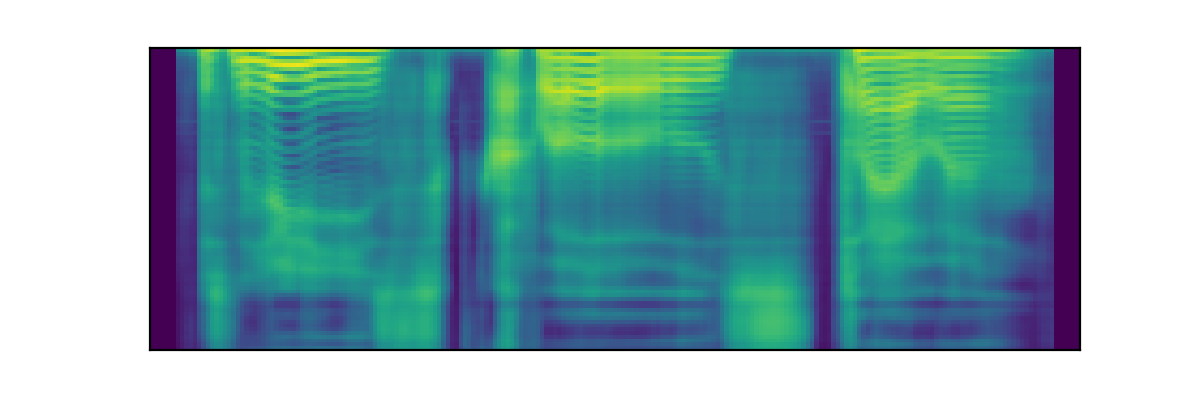}
		\caption{$\tau=1.08$}\label{subfig:1.08}
	\end{subfigure}
	\begin{subfigure}{\columnwidth}
		\centering
		\includegraphics[width=0.9\linewidth]{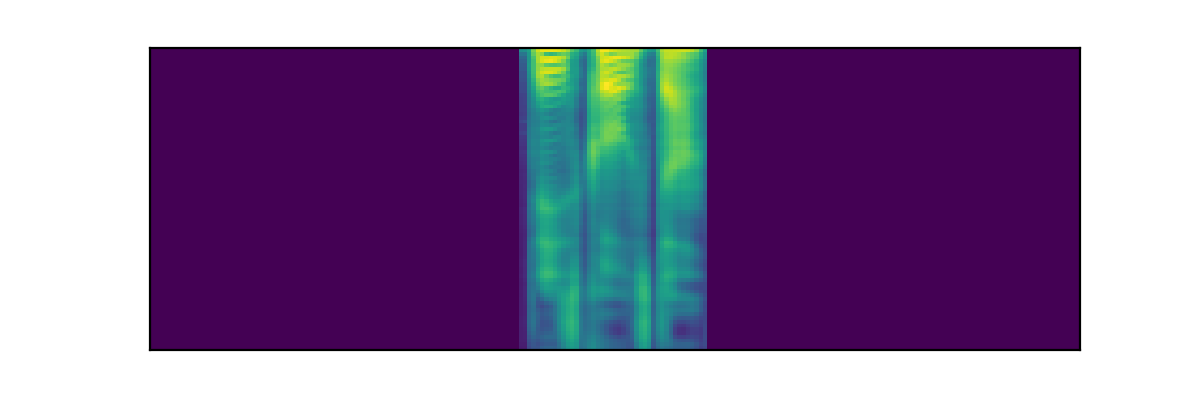}
		\caption{$\tau=0.90$}\label{subfig:0.9}
	\end{subfigure}
	\begin{subfigure}{\columnwidth}
		\centering
		\includegraphics[width=0.9\linewidth]{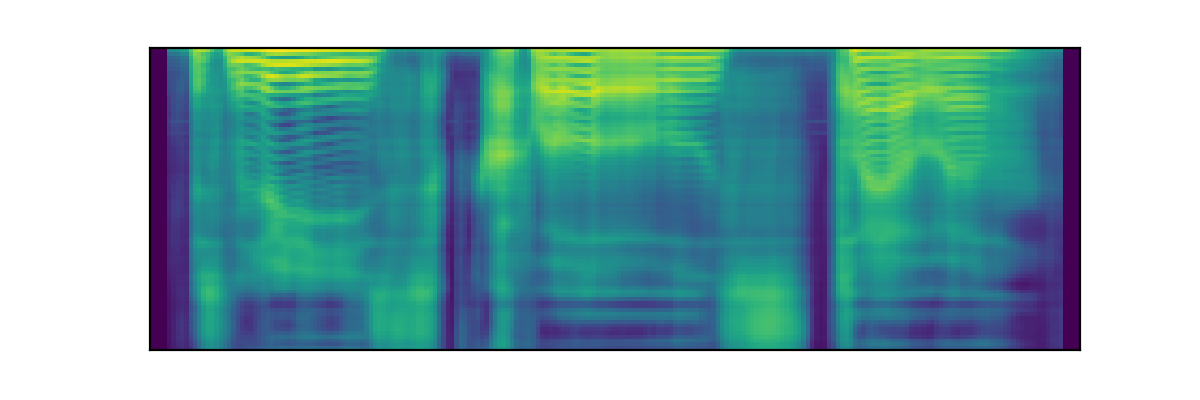}
		\caption{$\tau=1.10$}\label{subfig:1.1}
	\end{subfigure}
	\caption{Temporally-aligned spectrograms constructed from resampled hidden representation as a visualization of the resampling operations. The utterance is \emph{Please call Stella}.}
	\label{fig:resample_vis}
\end{figure*}




\end{document}